\documentclass[prx,twocolumn,superscriptaddress,showpacs,notitlepage,floatfix,bibnotes,nofootinbib]{revtex4-2}

\usepackage{lmodern}
\usepackage{microtype}

\usepackage{amsmath,tikz,amssymb,amsthm,mathtools,bm}
\usepackage{mathrsfs}   
\usetikzlibrary{arrows.meta,decorations.pathreplacing}

\usepackage{bbm} 
\usepackage{graphicx}
\usepackage{booktabs}
\usepackage{array}
\usepackage{thmtools}
\usepackage{thm-restate}

\usepackage[dvipsnames]{xcolor} %
\usepackage{hyperref}
\definecolor{green}{RGB}{11, 218, 81}

\hypersetup{colorlinks=true,allcolors=green,pdfborder={0 0 0}}
\theoremstyle{plain}
\newtheorem{theorem}{Theorem}
\newtheorem{lemma}{Lemma}
\newtheorem{proposition}{Proposition}
\newtheorem{corollary}{Corollary}
\newtheorem{assumption}{Assumption}

\theoremstyle{definition}
\newtheorem{definition}{Definition}

\theoremstyle{remark}

\newcommand{\C}{\mathbb{C}}

\newcommand{\eps}{\varepsilon}
\newcommand{\dd}{\mathrm{d}}

\newcommand{\OPT}{\mathsf{OPT}}
\newcommand{\poly}{\mathrm{poly}}
\newcommand{\tr}{\operatorname{tr}}
\newcommand{\Span}{\operatorname{span}}

\DeclareMathOperator{\Tr}{Tr}

\newcommand{\ket}[1]{|#1\rangle}
\newcommand{\bra}[1]{\langle #1|}

\newcommand{\norm}[1]{\left\|#1\right\|}

\renewcommand\vec{\boldsymbol}

\def\p{{\vec{p}}}

\begin{document}

\title{Learning the closest Slater determinant}
\author{Nisarga Paul}
\email{npaul@caltech.edu}
\affiliation{Department of Physics and Institute for Quantum Information and Matter,
California Institute of Technology, Pasadena, California 91125, USA}
\author{Haimeng Zhao}
\email{haimeng@caltech.edu}
\affiliation{Department of Physics and Institute for Quantum Information and Matter,
California Institute of Technology, Pasadena, California 91125, USA}
\author{David D. Dai}
\email{dddai@mit.edu}
\affiliation{Department of Physics, Massachusetts Institute of Technology, Cambridge, Massachusetts 02139, USA}
\affiliation{The NSF AI Institute for Artificial Intelligence and Fundamental Interactions}
  
\begin{abstract}
\noindent Learning compact, interpretable descriptions of quantum many-body states is an important task in quantum science. We study the task of learning the Slater determinant with maximum fidelity to an arbitrary fermionic many-body state, with motivation from both Hartree-Fock methods and agnostic tomography. Given an $n$-fermion wavefunction built from $m$ fermionic modes, we provide classical and quantum algorithms returning a Slater determinant with fidelity within $\varepsilon$ of maximal in time $m^{\text{poly}(n,1/\varepsilon)}$. We prove matching hardness lower bounds, assuming standard complexity conjectures, along some parameter axes. Given access to quantum copies, we prove this can be accomplished with $\text{poly}(m,n,1/\varepsilon)$ copies of $\rho$. We also show that above a fidelity of $2/3$ any stationary point is the unique global maximum while below $2/3$ the optimization landscape can have spurious stationary points, and hence $2/3$ marks a transition point in the optimization landscape for this problem. We apply the algorithm to the Fermi–Hubbard model, extracting the closest Slater determinant from neural quantum state solutions. Together, our results provide algorithmic tools with provable guarantees in understanding fermionic many-body systems with classical or quantum simulation.
\end{abstract}

\maketitle

\section{Introduction}

Extracting compact, interpretable descriptions of complex many-body wavefunctions is one of the central challenges of quantum science. For fermionic systems, the simplest such description is a Slater determinant, the antisymmetrized product of single-particle states~\cite{slater}. Despite their simplicity, Slater determinants can capture a remarkable range of quantum many-body phenomena, from molecules and nuclei to ferromagnets and topological phases~\cite{bruus2004many,girvin2019modern}. Slater determinants form the variational foundation of Hartree-Fock theory and density functional theory---two algorithmic workhorses in condensed matter physics and quantum chemistry---and are the reference around which beyond-mean-field methods such as coupled-cluster~\cite{RevModPhys.35.496,echenique2007mathematical,kohn1999nobel,shavitt2009many} expand. 

A natural question is then the following: \textit{given a fermionic many-body wavefunction, how does one find the closest Slater determinant?} We will quantify closeness by having a high fidelity. Despite its apparent simplicity, this question is surprisingly subtle and lacks a complete answer, since the manifold of Slater determinants is nonlinear and it is an essentially non-convex optimization problem (Fig.~\ref{fig:1}). Prior attempts have involved, for instance, heuristic numerical gradient ascent~\cite{zhang2014optimal,zhang2016optimal}, whose performance is not guaranteed; here, we seek algorithms that provably solve the problem with rigorous performance guarantees.

\begin{figure}
    \centering
    \includegraphics[width=\linewidth]{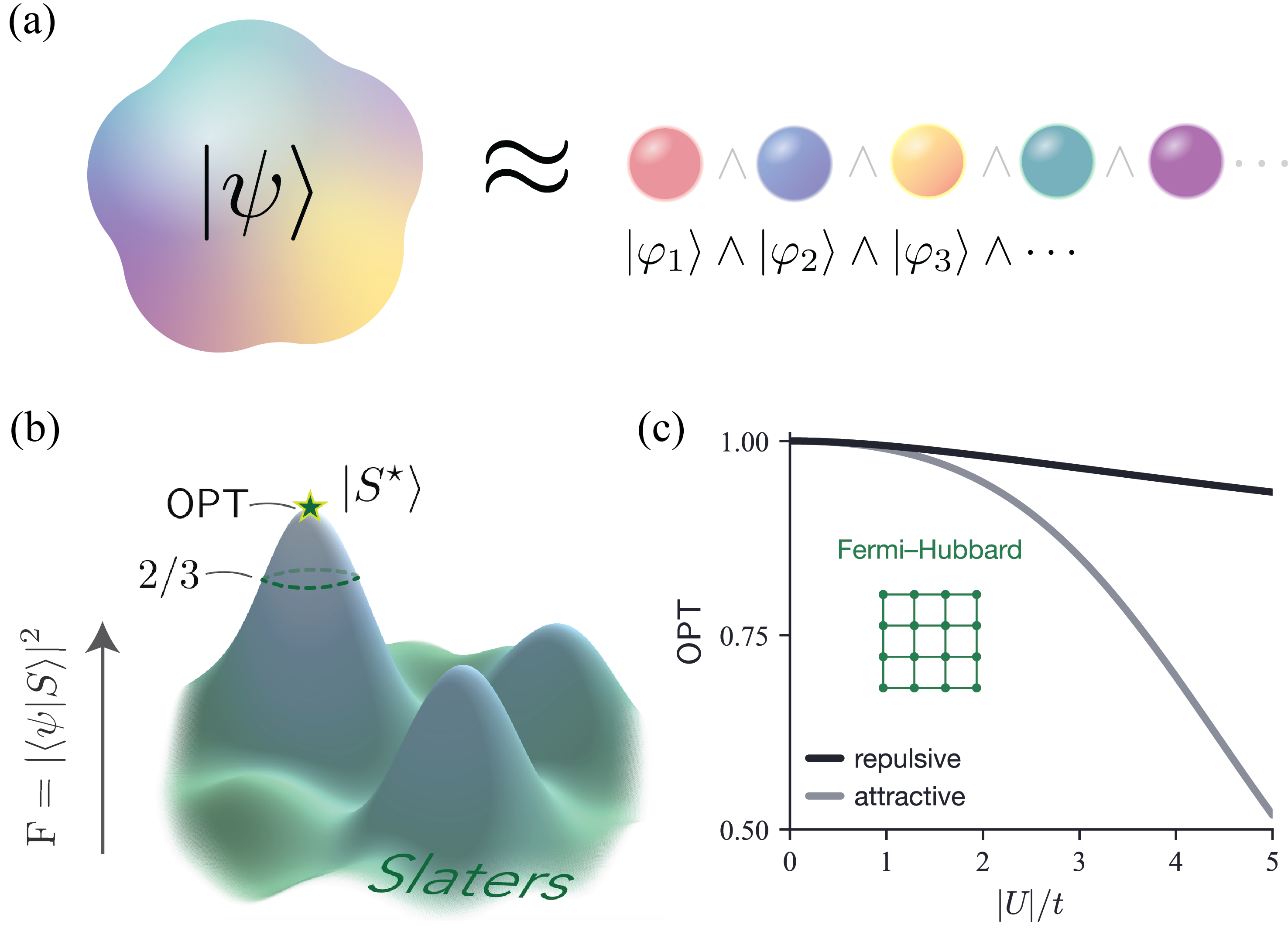}
\caption{%
\textbf{Learning the closest Slater determinant.}
(a) We provide a deterministic protocol to learn the closest Slater determinant to an arbitrary fermionic state $\ket{\psi}$. This allows the extraction of useful information about $\ket{\psi}$, such as its natural single-particle orbital basis.
(b)~Maximizing fidelity $\mathrm{F}$ over all Slater determinants is generally a non-convex optimization problem. The optimization landscape exhibits a transition at $\mathrm{F}=2/3$ above which there are no spurious stationary points. (c) The closest-Slater fidelity $\OPT$ of the $4\times4$ Fermi–Hubbard ground state versus interaction $|U|/t$, for repulsive and attractive coupling.}
\label{fig:1}
\end{figure}

\par 

A precise formulation of this problem is as follows.  Let $\mathcal{H}_n = \wedge^n \mathcal{H}_1$ be the Hilbert space of $n$-particle fermionic wavefunctions built from single-particle Hilbert space $\mathcal{H}_1$, with dimension $m = \dim \mathcal{H}_1$ and $ n \leq m$, and let
\begin{equation}
	\mathcal{S} = \left\{ \bigwedge_{i=1}^n \ket{\varphi_i}: \ket{\varphi_i} \in \mathcal{H}_1, \langle \varphi_i|\varphi_j\rangle = \delta_{ij}\right\}
\end{equation}
be the set of Slater determinants. Given a state in $\mathcal{H}_n$ with density matrix $\rho$ and an error tolerance $\varepsilon >0$, we seek $\ket{S}\in \mathcal{S}$ such that 
    \begin{equation}\label{eq:OPT}
    \langle S|\rho |S\rangle \geq \OPT - \varepsilon,\quad \text{where } \OPT = \max_{\ket{S}\in \mathcal{S}}\langle S|\rho |S\rangle.
\end{equation}
Here, $\rho$ can be given to us either as a classical description, for applications in analyzing tensor networks or neural quantum states, or as copies of a quantum state produced by a quantum device or quantum simulator. We will be interested in the complexity of a classical or quantum algorithm solving this problem. \par 

This task is fundamental in physical sciences, where we advance our understanding of the physical world by learning simple ans\"atze to describe complicated quantum states, with the approximation accuracy in mind. This is also of practical interest in both condensed matter physics and quantum technology. On the condensed matter side, this problem is related to (but distinct from) the Hartree-Fock (HF) approximation, which minimizes the energy expectation value of a Hamiltonian $H$ over $\mathcal{S}$. Here we are maximizing fidelity with a given state $\rho$, with no reference to a Hamiltonian. In general, the two problems need not coincide unless $\rho$ is the ground state of $H$ and $\rho \in \mathcal{S}$. For a numerical practitioner with some compact description of a fermionic state, such as a matrix product state or neural quantum state~\cite{PhysRevLett.69.2863,schollwock2011density,carleo2017solving}, finding the best Slater determinant approximation allows the extraction of meaningful information such as a natural single-particle occupation basis and notion of particle-hole excitations~\cite{PhysRev.97.1474,roos1980complete,Thouless1960,shavitt2009many}. On the quantum technology side, algorithms solving this problem can be used to benchmark and understand the outcomes of fermionic quantum simulators and computers, which fall into the broader task of \textit{agnostic tomography}~\cite{grewal2026agnostic,chen2025stabilizer,zhao2024learning,wadhwa2025agnostic}, where the goal is to find a state in some class $\mathcal{C}$ maximizing overlap with $\rho$, without the assurance that $\rho\in \mathcal{C}$. In particular, our problem is the natural fermionic analog of closest product state learning~\cite{chen2025stabilizer,bakshi2025learning,wei2003geometric,vedral1997quantifying},
but the antisymmetry of fermionic wavefunctions renders it structurally quite different. Finally, for an experimentalist with access to copies of $\rho$, a practical solution to this problem yields a useful proxy for an unknown, potentially complex many-body state.

 \par

\textbf{Results. }Our main results are twofold:
\begin{itemize}
\item \textbf{Learning algorithms. }
We give a classical and a quantum algorithm to find the closest Slater determinant both with runtime $m^{\poly(n, 1/\varepsilon)}$, polynomial in $m$ for constant $n$ and $\varepsilon$. The quantum algorithm has a $\poly(m,n, 1/\varepsilon)$ sample complexity.  The algorithms truncate $\rho$ to a subspace of ``natural orbitals" and perform a discretized search over the corresponding subspace of Slater determinants. Under a standard complexity conjecture, no classical algorithm can substantially improve the runtime dependence in $n$. Moreover, no classical or quantum algorithm that runs in time polynomial in $1/\varepsilon$ can solve the problem unless $\mathsf{NP}\subseteq \mathsf{BPP}/ \mathsf{BQP}$. 
\item \textbf{Optimization landscape.} 
We reveal a sharp, qualitative transition in the optimization landscape of finding the closest Slater determinant.
In particular, we show that if a Slater determinant $|S\rangle$ is a stationary point of the fidelity and has fidelity $\mathrm{F}=\langle S | \rho | S \rangle > 2/3$, then it is the global unique closest Slater determinant. 
In other words, above the threshold of $2/3$, the fidelity optimization landscape is benign in the sense of lacking spurious stationary points, a strong guarantee in a non-convex optimization problem. Below $\mathrm{F}=2/3$, this guarantee can always be thwarted, and hence $2/3$ marks a sharp transition in the nature of the optimization landscape (Fig.~\ref{fig:1}b).
\end{itemize}

We now describe these results in more detail. We separate the case of classical access to $\rho$, for which we assume an efficient way to compute amplitudes, and quantum access to $\rho$, in which case we assume we can make measurements on quantum copies of $\rho$. In either case, the algorithms proceed by estimating the single-particle reduced density matrix (1-RDM) and isolating a $\mathrm{poly}(n,1/\varepsilon)$-dimensional subspace of $\mathcal{H}_1$ comprising high-occupancy orbitals. We then discretize the space of Slater determinants comprising these orbitals and find the one with maximum fidelity. This procedure takes $m^{\mathrm{poly}(n,1/\varepsilon)}$ time and, using quantum threshold search~\cite{buadescu2021improved}, consumes a polynomial number of samples in the quantum case. 
We prove that this time complexity is essentially optimal in $1/\varepsilon$ for both the quantum and classical algorithms, and essentially optimal in $n$ for the classical algorithm, under widely-believed complexity-theoretic assumptions. To establish the latter, we prove that an $m^{o(n)}$-time classical algorithm could also solve the $n$-\textsc{Multicolored Clique} problem in an unacceptably small time~\cite{CyganEtAl2015}. For the former, we prove that a $\poly(m,1/\varepsilon)$-time algorithm in either the classical or quantum setting would efficiently solve the $\mathsf{NP}$-hard quantum separability problem~\cite{gharibian2010strong}. We summarize the state of affairs in Table~\ref{tab:results}. 

\begin{table}[t]
\begin{ruledtabular}
\begin{tabular}{lcc}
\multicolumn{3}{l}{\textit{\textbf{Upper bounds}} (algorithms)}\\
classical runtime & $m^{\poly(n,1/\varepsilon)}$ & Thm.~\ref{thm:upper}\\
quantum runtime & $m^{\poly(n,1/\varepsilon)}$ & Thm.~\ref{thm:quantum}\\
quantum sample complexity & $\poly(m,n,1/\varepsilon)$ & Thm.~\ref{thm:quantum}\\
\colrule
\multicolumn{3}{l}{\textit{\textbf{Lower bounds}} (hardness)}\\
runtime in $n$ (classical) & $\ge m^{\Omega(n)}$ & Thm.~\ref{thm:lower}\\
runtime in $1/\varepsilon$ & $\ge (1/\varepsilon)^{\omega(1)}$ & Thm.~\ref{thm:comp-hard}\\
\end{tabular}
\end{ruledtabular}
\caption{\textbf{Summary of complexity results.} Upper bounds are achieved by our algorithms; lower bounds hold under standard complexity-theoretic conjectures. Here $n$ is the particle number, $m$ the number of single-particle modes, and $\varepsilon$ the additive error in fidelity.}
\label{tab:results}
\end{table}

Our second main result shows that the geometry of the optimization landscape exhibits a sharp, qualitative transition at fidelity $2/3$. 
In particular, in Section~\ref{sec:highfid} we show that spurious stationary points (e.g. local optima) of the fidelity function are forbidden above $\mathrm{F}=2/3$. Moreover, this threshold is sharp in the sense that there exist pure states $\ket\psi$ whose fidelity has stationary points with fidelity arbitrarily close to $2/3$ from below. The definition of stationarity is provided in Section~\ref{sec:highfid}, which closely resembles the Brillouin condition from Hartree-Fock theory~\cite{Brillouin1933}. Any Slater determinant $\ket{S}$ satisfying this condition and $\mathrm{F}>2/3$ is hence guaranteed to be the global optimum, though there may be other stationary points with $\mathrm{F}<2/3$. This is a useful guarantee for an optimization problem otherwise only amenable to a brute-force solution: a heuristic algorithm such as gradient ascent or Hartree-Fock-style iteration, once it reaches a fidelity above $2/3$, should converge to the global optimum.

\par
We complement these results with numerical experiments, described in
Sec.~\ref{sec:numerics}. We apply our classical algorithm to ground states of the
Fermi--Hubbard model, extracting the closest Slater determinant directly from exact diagonalization (Fig.~\ref{fig:1}c) or a
neural quantum state representation (Fig.~\ref{fig:numerics}a). Using $\OPT$ as
the ground truth, we find that simple gradient ascent can be unreliable as system grows, and we illustrate the distribution of stationary points within the optimization landscape (Fig.~\ref{fig:numerics}b,c).

\textbf{Related Work. } The problem of finding the closest Slater determinant to a given state has been studied empirically using gradient ascent algorithms~\cite{PhysRevA.102.052803,zhang2014optimal,zhang2016optimal}. These, however, hold no guarantees of convergence to the global optimum, which is our focus. Refs. \cite{aaronson_et_al:LIPIcs.TQC.2023.12,o2022fermionic} give $\mathrm{poly}(n,m,1/\varepsilon)$-efficient solutions to this problem when $\rho$ is known to be a Slater determinant, \textit{i.e.} $\mathsf{OPT} =1$. In this setting, the state is entirely specified by its 1-RDM, which can be estimated efficiently~\cite{huang2020predicting,zhao2021fermionic,bittel2025optimal,christensen2026learning}. However, this approach does not work here because the 1-RDM by itself does not determine the closest Slater determinant in general\footnote{Consider $\frac{c_1^\dagger c_2^\dagger + c_3^\dagger c_4^\dagger}{\sqrt{2}}\ket{0}$ and $\frac{c_1^\dagger c_3^\dagger + c_2^\dagger c_4^\dagger}{\sqrt{2}}\ket{0}$, which have identical 1-RDMs $\propto I$ but different closest Slaters, namely the pairs $\{c_1^\dagger c_2^\dagger\ket{0},c_3^\dagger c_4^\dagger \ket{0}\}$ and $\{c_1^\dagger c_3^\dagger\ket{0},c_2^\dagger c_4^\dagger\ket{0}\}$, respectively.}. Indeed, hardness results on learning even \textit{free} fermion states from measurements~\cite{bittel2025pac,bittel2025optimal} serve as precautions for this agnostic tomography problem. We note that several works have proposed using quantities like $\OPT$ to capture the amount of correlations in a fermionic state~\cite{gottlieb2005new,gottlieb2007properties,benatti2012entanglement,turner2017optimal}. Indeed, there is a significant body of literature on correlation measures of fermionic states using non-Gaussianity, non-stabilizerness, occupation number entropies, and more~\cite{hebenstreit2019all,dias2024classical,lyu2024fermionic,coffman2025measuring,ares2026asymmetry,sierant2026fermionic,leone2022stabilizer,gigena2015entanglement,vanhala2024complexity,tarabunga2026computablemeasuresfermionicnongaussianity}.\par 

The closest analog of our problem in the bosonic setting is the problem of learning the closest product state to a given $n$-qubit state, recently studied by Bakshi et al. in the setting of quantum copy access~\cite{bakshi2025learning}. Bakshi et al. give an algorithm with sample and time complexity $n^{\mathrm{poly}(1/\varepsilon)}$ to achieve this task within additive error $\varepsilon$, and also prove computational hardness when the desired accuracy is inverse-polynomial in $n$. A product state is specified by independent local factors
$|\pi_1\rangle\otimes\cdots\otimes|\pi_n\rangle$, and the algorithm of
Bakshi et al. exploits this structure with an algorithm proceeding qubit by qubit.  This algorithm does not transfer directly here, since Slaters do not factorize over particles.
\par

\par 
\textbf{Outline. } The remainder of this paper is organized as follows. In
Section~\ref{sec:algorithm}, we present the classical algorithm for learning the closest Slater determinant. In
Section~\ref{sec:quantum}, we give a quantum protocol for the problem. In Section~\ref{sec:hardness}, we present computational hardness lower bounds in both the classical and quantum settings. We study the properties of the optimization landscape in Section~\ref{sec:highfid}, report numerical results in
Section~\ref{sec:numerics}, and conclude in Section~\ref{sec:discussion}.
\par

\section{Algorithms for learning the closest Slater}
\label{sec:algorithm}

\begin{figure*}
    \centering
    \includegraphics[width=0.95\linewidth]{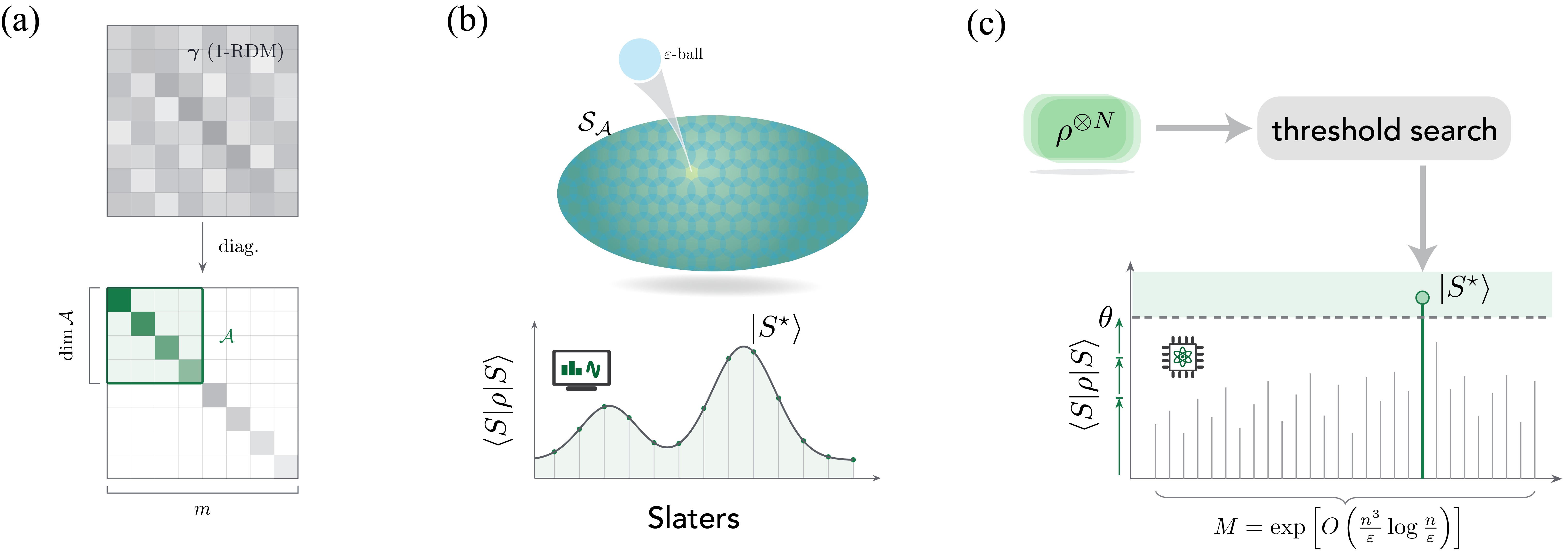}
\caption{\textbf{Overview of algorithms.}
(a) The $m\times m$ 1-RDM $\gamma$ is diagonalized into natural orbitals, and an active space $\mathcal{A}$ of large occupation (dimension
$\dim\mathcal{A}=O(n^2/\varepsilon)$) is retained.
(b) An $\varepsilon$-covering net over the manifold
$\mathcal{S}_{\mathcal{A}}$ of Slater determinants built from $\mathcal{A}$ is followed by explicit search in the classical case (Theorem~\ref{thm:upper}). The net has size $M=\exp[O(\tfrac{n^3}{\varepsilon}\log\tfrac{n}{\varepsilon})]$. (c) In the quantum setting (Theorem~\ref{thm:quantum}), quantum copies
$\rho^{\otimes N}$ are fed to the threshold search routine of
Ref.~\cite{buadescu2021improved}, applied to the same $\varepsilon$-covering net. The closest Slater $\ket{S^\star}$ can be found using $N=\poly(n,1/\varepsilon)$ copies of $\rho$.}
    \label{fig:algorithm}
\end{figure*}

We first establish some notation. We take $\mathcal{H}_n = \wedge^n \mathcal{H}_1$ to be the Hilbert space of $n$-particle fermionic wavefunctions where $\mathcal{H}_1$ is the single-particle Hilbert space. We set $m=\dim \mathcal{H}_1$ and $n\leq m$. Then 
\begin{equation}
    \mathcal{S} = \left\{\bigwedge_{i=1}^n \ket{\varphi_i} : \ket{\varphi_i} \in \mathcal{H}_1, \langle \varphi_i|\varphi_j\rangle = \delta_{ij}\right\}
\end{equation}
is the set of Slater determinants (or Slaters for short). We denote Slaters by capital letters, \textit{e.g.}, $\ket{S}$, $\ket{T}$. Let $\rho$ be an arbitrary fermionic state, either pure or mixed. When $\rho$ is pure, we write $ \rho = \ket{\psi}\!\bra{\psi}$ where $\ket{\psi} \in \mathcal{H}_n$. We center the main text around pure states to lighten exposition, generalizing to the mixed state case in Appendix~\ref{app:mixed}. We will optimize the fidelity,
\begin{equation}
    \text{F}(S) = \langle S|\rho|S\rangle,
\end{equation}
where we assume unit normalization $\Tr \rho =\langle S|S\rangle = 1$. \par 

Let $\ket{0}$ be the zero particle vacuum. Let $c_1^\dagger,\ldots, c_m^\dagger$ create an orthonormal basis $|\varphi_1\rangle,\ldots, |\varphi_m\rangle$ in $\mathcal{H}_1$. A basis for $\mathcal{H}_n$ is the set of states
\begin{equation}
    \ket{I} = c_{i_1}^\dagger c_{i_2}^\dagger \cdots c_{i_n}^\dagger \ket{0}
\end{equation}
where $I = \{i_1<i_2< \cdots < i_n\} \subset \{1,\ldots, m\}$. A general Slater determinant can be expressed in terms of an $m\times n$ matrix $U$ with $U^\dagger U = 1$ via 
\begin{equation}
    \ket{S} = d_1^\dagger \cdots d_n^\dagger \ket{0},\qquad d_i^\dagger = \sum_{j=1}^m U_{ji}c_j^\dagger.
\end{equation}
The amplitudes of $\ket{S}$ are given by $\langle I|S\rangle = \det U[I,:]$ where $U[I,:]$ is the $n\times n$ submatrix on rows $I$. \par 

In the classical setting, we will assume efficient query access to the amplitudes $\langle S|\psi\rangle$ of the pure state $\ket\psi$ (leaving the mixed state access model to Appendix~\ref{app:mixed}). We will also assume few-body observables such as $\Tr(c_i^\dagger c_j \rho)$ are efficiently estimable: that is, we can efficiently sample a random variable $X_O$ such that $\mathbb{E}[X_O] = \Tr(O\rho)$ and $\mathrm{Var}[X_O]$ is constant and bounded, for few-body observables $O$. This is a model for a succinct classical description such as a tensor network or neural quantum state with possibly stochastic estimation. 
In the quantum setting, we will assume access to quantum copies of the state $\rho$ and the ability to perform arbitrary measurements. In both the classical and quantum settings, the first step will be to estimate the 1-RDM 
\begin{equation}
    \gamma_{ij} = \Tr( c_j^\dagger c_i \rho),
\end{equation}
which is an $m\times m$ matrix. For either a classical noisy estimator or a quantum tomography protocol, this can be achieved within operator norm error $\eta$ with $\mathrm{poly}(m,1/\eta)$ runtime or measurements~\cite{zhao2021fermionic,christensen2026learning}. We will find that $\eta = O(\varepsilon/n)$ suffices for later purposes, where $\varepsilon$ is the final additive error in Eq.~\eqref{eq:OPT}, so $\mathrm{poly}(m,n,1/\varepsilon)$ measurements suffice overall.
\par 

We begin by describing the classical algorithm for learning the closest Slater, which has three stages: estimating the 1-RDM, diagonalizing it and restricting to an \textit{active space}, and searching over the active space Slaters. For the first stage, writing $\hat\gamma$ as the estimator for $\gamma$, we perform an orthonormal eigendecomposition 
\begin{equation}
    \hat\gamma_{ij} = \sum_{a=1}^m \hat\nu_a (\hat\psi_i^a)^* \hat\psi_j^a,
\end{equation}
which gives us an approximate natural basis $\hat\psi^a$ for $\mathcal{H}_1$. The eigenvalues $0\leq \hat\nu_a\leq 1$ are occupations of the natural orbitals $\hat\psi^a$. \par
The next step is to truncate to $\hat\psi^a$ with nontrivial occupations $\hat\nu_a$ to isolate an active space $\mathcal{A}\subset \mathcal{H}_1$. This step mirrors the routine practice of finding an active space in quantum chemistry and electronic structure calculations~\cite{veryazov2011select,stein2016automated,sayfutyarova2017automated,ma2013assessment}, though it is simpler here. In particular, we choose $\mathcal{A}$ to be
\begin{equation}
\mathcal{A} = \mathrm{span}\{ \hat\psi^a \mid \hat\nu_a > 3\varepsilon/8n\}.
\end{equation}
If $\dim\mathcal A<n$, we enlarge $\mathcal A$ by adding arbitrary orthonormal
orbitals until $\dim\mathcal A=n$ without loss of generality. In Lemmas \ref{lem:active} and \ref{lem:noisy-active} in the Appendix, we prove that it suffices to restrict our attention to $\mathcal{A}$ instead of the entire space $\mathcal{H}_1$, incurring at most an $\varepsilon/2$ additive error in the objective $\OPT$, leaving $\varepsilon/2$ for the next step. We illustrate this schematically in Fig.~\ref{fig:algorithm}a. If $\|\hat\gamma-\gamma\|\le \varepsilon/8n$, Lemma~\ref{lem:noisy-active}
implies~\footnote{before padding, and the same holds after padding for $0<\varepsilon\le1$.}
\begin{equation}
    \dim\mathcal A
    \le
    \min\left(m,\frac{4n^2}{\varepsilon}\right),
\end{equation}
 which is advantageous when $n^2/\varepsilon \ll m$. If this condition fails to hold, the truncation to $\mathcal{A}$ does not result in appreciable complexity savings and is not necessary. We'll assume $n^2/\varepsilon\ll m$ from here on, and thus $\dim\mathcal A=O(n^2/\varepsilon)$. It remains to optimize the fidelity over this space, which we accomplish below (with the mixed-state generalization in Appendix~\ref{app:mixed}). 
\begin{theorem}[Learning the closest Slater determinant, the classical setting]\label{thm:upper}
Let $\rho=\ket{\psi}\!\bra{\psi}$ be an $n$-particle fermionic state. Assume we have an estimator $\hat\gamma$ for the 1-RDM $\gamma_{ij} = \Tr(c_j^\dagger c_i\rho)$ satisfying $\|\hat\gamma - \gamma\|<\varepsilon/8n$, together with classical query access to the amplitudes $\langle S|\psi\rangle$. Then we can find the closest Slater to $\rho$, as measured by fidelity, to an additive error $\varepsilon$ in time $\exp\!\left[
    O\!\left(
      \frac{n^3}{\varepsilon}\log\frac{n}{\varepsilon}
    \right)
  \right]$
on a classical computer.
\end{theorem}

\begin{proof}
By Lemmas \ref{lem:active} and \ref{lem:noisy-active}, we may restrict the search to Slaters whose occupied orbitals lie in $\mathcal{A}$ at the cost of an additive $\varepsilon/2$ error in the fidelity:
\begin{equation}\label{eq:OPTleqOPTA}
    \OPT \leq \OPT_{\mathcal{A}} + \varepsilon/2, \quad \OPT_{\mathcal{A}} = \max_{\ket S \in \mathcal{S}_\mathcal{A}}\langle S|\rho|S\rangle.
\end{equation}
Here, $\mathcal{S}_\mathcal{A}$ is the manifold of Slater determinants built from orbitals in $\mathcal{A}$. We will bound $\dim\mathcal{A}$ and then consider a covering net over $\mathcal{S}_\mathcal{A}$.
\par 
We recall that the active subspace $\mathcal{A}$ consists of eigenvectors of $\hat\gamma$ with eigenvalues $> 3\varepsilon/8n$. For any $\ket{u}\in \mathcal{A}$, we have $\langle u|\gamma|u\rangle \geq \langle u|\hat\gamma|u\rangle -\|\hat\gamma-\gamma\| = 3\varepsilon/8n - \varepsilon/8n = \varepsilon/4n$. Summing over an orthonormal basis of $\mathcal{A}$ gives $\tr \gamma \geq r\varepsilon/4n$ where $r\coloneqq \dim\mathcal{A}$, and hence $r< 4n^2/\varepsilon =O(n^2/\varepsilon)$.
\par 

We will parameterize each $\ket{S}\in\mathcal{S}_{\mathcal{A}}$ by an $r\times n$ matrix $U$ with orthonormal columns. We denote this as $\ket{S(U)}$ and the fidelity as
\begin{equation}
    \mathrm{F}(U) = \langle S(U)|\rho|S(U)\rangle.
\end{equation}
It will be necessary to bound $|\mathrm{F}(U)-\mathrm{F}(V)|$ in terms of the difference $\|U-V\|_F$ (in Frobenius norm) for two matrices $U,V$. One can show~\footnote{Letting $A = \ket{S(U)}, B = \ket{S(V)}$, we have $|\mathrm{F}(U)-\mathrm{F}(V)| = |\langle A-B|\rho|A\rangle +\langle B|\rho|A-B\rangle|\leq \|A-B\|\|\rho\|_{\mathrm{op}}\|A\|+\|B\|\|\rho\|_{\mathrm{op}}\|A-B\| =2\|A-B\|$, since $\|\rho\|_{\mathrm{op}}\leq 1$ and $\|A\|=\|B\|=1$.}
\begin{equation}\label{eq:FUFV}
    |\mathrm{F}(U)-\mathrm{F}(V)| 
    \leq 2\|\ket{S(U)}-\ket{S(V)}\|.
\end{equation}
We must next bound $\|\ket{S(U)}-\ket{S(V)}\|$ in terms of $\|U-V\|_F$. We write
\begin{equation}
    \ket{S(U)} = \ket{u_1}\wedge \cdots \wedge \ket{u_n}
\end{equation}
where $\ket{u_k}$ corresponds to the $k$'th column of $U$ viewed as a single-particle orbital, and similarly for $\ket{S(V)}$. We claim 
\begin{equation}\label{eq:SUSVuv}
    \|\ket{S(U)}-\ket{S(V)}\| \leq \sum_{k=1}^n \| \ket{u_k}-\ket{v_k}\|
\end{equation}
and establish this by induction on $n$. For $n=1$, this is true as an equation. For $n\geq 2$, we assume it is true for $n-1$ and we add and subtract $\left[\bigwedge_{k=1}^{n-1}\ket{v_k}\right]\wedge \ket{u_n}$, yielding
\begin{multline}
    \ket{S(U)}-\ket{S(V)} = \left[\bigwedge_{k=1}^{n-1}\ket{u_k} -\bigwedge_{k=1}^{n-1}\ket{v_k} \right] \wedge \ket{u_n}\\
    +\left[\bigwedge_{k=1}^{n-1}\ket{v_k}\right]\wedge (\ket{u_n}-\ket{v_n}).
\end{multline}
Taking norms and applying the triangle inequality, we have
\begin{align}
    \| \ket{S(U)}-\ket{S(V)} \| &\leq \left\|\bigwedge_{k=1}^{n-1}\ket{u_k} -\bigwedge_{k=1}^{n-1}\ket{v_k}\right\|\|\ket{u_n}\|\nonumber\\
&+\left\|\bigwedge_{k=1}^{n-1}\ket{v_k}\right\| \|\ket{u_n}-\ket{v_n}\|\\
    &\leq \sum_{k=1}^{n-1}\|\ket{u_k}-\ket{v_k}\|+ \|\ket{u_n}-\ket{v_n}\|\nonumber
\end{align}
where we used the inductive hypothesis and $\|\wedge_{k=1}^{n-1}\ket{u_k}\| \leq \prod_{k=1}^{n-1}\|\ket{u_k}\| = 1$. This establishes Eq.~\eqref{eq:SUSVuv}. Next, we note that $\|U-V\|_F^2 = \sum_{k=1}^n \|\ket{u_k}-\ket{v_k}\|^2$, which implies $\|\ket{u_k}-\ket{v_k}\|\leq \|U-V\|_F$ for each $k$. Combining from Eq.~\eqref{eq:FUFV}, we have
\begin{equation}
    |\mathrm{F}(U)-\mathrm{F}(V)| \leq 2n\|U-V\|_F. 
\end{equation}

Next, we consider a covering of $\mathcal{S}_{\mathcal{A}}$, which is a compact manifold of real dimension $d=2n(r-n)$, as illustrated heuristically in Fig.~\ref{fig:algorithm}b. It hence admits a $\delta$-covering net with $(c/\delta)^d$ elements for some constant $c$~\cite{szarek1982nets}. We take $\delta = \varepsilon/4n$ so that every $\ket{S(U)}$ has a net point $\ket{S(U')}$ such that $\|U-U'\|_F <\varepsilon/4n$ and hence $|\mathrm{F}(U)-\mathrm{F}(U')| \leq \varepsilon/2$. The net has number of elements
\begin{equation}\label{eq:2nc}
    \left(\frac{4nc}{\varepsilon}\right)^{2n(r-n)} = \exp\left[O\left(\frac{n^3}{\varepsilon}\log \frac{n}{\varepsilon}\right)\right]
\end{equation}
where we used $r = O(n^2/\varepsilon)$. We then enumerate this net, evaluate F at each point, and return a Slater with maximum fidelity among all net points. Let $\ket{S(U')}$ be the corresponding Slater and $\ket{S(U^\star)}$ be a \textit{true} maximum over all $\mathcal{S}_{\mathcal{A}}$. Then
\begin{equation}
        \mathrm{F}(U') \geq \mathrm{F}(U^\star)-\varepsilon/2 = \OPT_{\mathcal{A}}-\varepsilon/2
    \geq \OPT -\varepsilon,
\end{equation}
using Eq.~\eqref{eq:OPTleqOPTA}. This procedure hence returns an optimal Slater to within an additive $\varepsilon$ tolerance, as desired. 
\end{proof}

It is worth clarifying why the covering net search of Theorem~\ref{thm:upper} may be
necessary in general. Having isolated the active space from the 1-RDM, one might hope to
avoid the covering net search and simply occupy the $n$ natural orbitals of
largest occupation, returning the resulting Slater. This approach succeeds when
$\rho$ is a Slater and may be accurate when the 1-RDM eigenvalues lie close to either $0$ or $1$,
but it can fail badly in general. In other words, the closest Slater is not a simple function of the 1-RDM alone. For example, consider the state
\begin{equation}
    \ket{\psi} = \frac{1}{\sqrt3}\left(
      c_1^\dagger c_2^\dagger c_3^\dagger
      + c_3^\dagger c_4^\dagger c_5^\dagger
      + c_5^\dagger c_6^\dagger c_1^\dagger\right)\ket0.
\end{equation}
Each pair of terms differs in two orbitals, so the 1-RDM is diagonal, with
occupation $2/3$ on orbitals $1,3,5$ and $1/3$ on orbitals $2,4,6$. The Slater corresponding to the three leading occupations in the 1-RDM is thus $c_1^\dagger c_3^\dagger c_5^\dagger\ket{0}$, which has zero overlap with $\ket{\psi}$, while the true closest Slater fidelity satisfies $\OPT\geq 1/3$ (taking e.g. $c_1^\dagger c_2^\dagger c_3^\dagger\ket{0}$). Indeed, there is no efficient (polynomial in $n,1/\varepsilon$) protocol for determining the closest Slater from the 1-RDM alone in general under certain computational complexity assumptions, as we detail in Sec.~\ref{sec:hardness}. 

\par 
In conclusion, the classical algorithm presented here gives an answer to the question posed at the beginning of the paper
in its precise form (Eq.~\eqref{eq:OPT}). We turn to the quantum setting next.

\section{Quantum setting}\label{sec:quantum}

We turn to the problem of learning the closest Slater under the assumptions of the quantum access model. In particular, we assume access to quantum copies of the state $\rho$ with the ability to make arbitrary measurements. We design an algorithm with $\poly(n,1/\varepsilon)$ sample complexity, independent of $m$, using the active space estimate; identifying the active space may incur a separate $\mathrm{poly}(m,n,1/\varepsilon)$ overhead, as we discuss. The runtime remains exponential in $\poly(n,1/\varepsilon)$. We leave open the question of whether a subexponential-in-$n$ runtime is possible. 
\par 
Our protocol makes use of the quantum threshold search (QTS) algorithm of Badescu and O’Donnell~\cite{buadescu2021improved}. Let us review the relevant parts. Let $\rho$ be an unknown quantum state and $A_1,\ldots, A_M$ observables satisfying $0\preceq A_i\preceq 1$~\footnote{where $A\preceq B$ means $B-A$ is positive semidefinite} and let $\theta_1,\ldots, \theta_M\in [0,1]$. Given a desired accuracy $\varepsilon>0$ and a tolerable failure probability $\delta > 0$, Ref.~\cite{buadescu2021improved} provides an algorithm using
\begin{equation}\label{eq:nQTS}
    n_{\mathrm{QTS}} = O\left(\frac{1}{\varepsilon^2}\left(\log^2M+ \log\frac{1}{\delta}\right)\log\frac{1}{\delta}\right)
\end{equation}
copies of $\rho$ which, except with probability at most $\delta$, either returns an index $j$ such that $\Tr(\rho A_j) > \theta_j-\varepsilon$ or certifies that $\Tr(\rho A_j)<\theta_j$ for all $j$. We will apply this primitive to the projectors onto a finite net of Slater determinants (c.f. Fig.~\ref{fig:algorithm}, Theorem~\ref{thm:upper}). The noteworthy feature is the $\log^2M$ dependence, which is exponentially better than the $\Omega(M)$ copies a na\"ive search over the $M$ candidates would require. 
\par 

We will need to slightly generalize the quantum threshold search algorithm to also search over the thresholds $\theta_1,\ldots, \theta_M$. In fact, we will set them to be all equal ($\theta = \theta_1 = \cdots =\theta_M$) and perform a binary search over $\theta \in [0,1]$. The final value of $\theta$ will satisfy $|\theta - \OPT|<\varepsilon$. This is achieved by the following Lemma, whose proof is left to Appendix~\ref{app:lemmas}. \par 
\noindent\textbf{Lemma \ref{lem:agnosticQTS}} (Agnostic quantum threshold search) \textit{Let $A_1,\ldots, A_M$ be known observables with $0\preceq A_i \preceq 1$ and let $F_i \coloneqq \Tr(\rho A_i)$ and $F_\star = \max_i F_i$. Given quantum copies of $\rho$, there is a procedure which outputs an index $i_\star$ such that $F_{i_\star} \geq F_\star -\varepsilon$ with probability at least $1-\delta$, using }
\begin{equation}\label{eq:naqts}
    n_{\mathrm{AQTS}} = O\left(\frac{\mu}{\varepsilon^2}\left(\log^2M + \log\frac{\mu}{\delta}\right) \log\frac{\mu}{\delta}\right)
\end{equation}
\textit{copies of $\rho$, where $\mu \coloneqq \log (1/\varepsilon)$. 
}
\par 

We will apply agnostic QTS to the problem of learning the closest Slater determinant in the quantum setting. This is presented formally in the following Theorem, depicted schematically in Fig.~\ref{fig:algorithm}c. \par 
\begin{theorem}[Learning the closest Slater determinant, the quantum setting]\label{thm:quantum}
Let $\rho$ be an arbitrary $n$-fermion state and suppose we have an estimator $\hat \gamma$ for the 1-RDM $\gamma$ satisfying $\|\hat\gamma-\gamma\|\leq \frac{\varepsilon}{12n}$. Suppose for each explicitly specified Slater determinant $\ket{S}$ we can implement the two-outcome measurement $\{\ket{S}\!\bra{S},1-\ket{S}\!\bra{S}\}$. Then there is a quantum protocol which outputs a Slater determinant $\ket{S^\star}$ satisfying $\langle S^\star|\rho|S^\star\rangle\geq \OPT-\varepsilon$ with probability at least $1-\delta$, using
\begin{equation}
  N=  \widetilde O\left(\frac{n^6}{\varepsilon^4}\right)
\end{equation}
copies of $\rho$, where $\widetilde O$ omits logarithmic factors in $n,1/\varepsilon,$ and $1/\delta$. The runtime is exponential in $\mathrm{poly}(n,1/\varepsilon)$. 
\end{theorem}

\begin{proof}
    From the 1-RDM estimator satisfying $\|\hat\gamma-\gamma\|\leq \frac{\varepsilon}{12n}$, we perform an eigendecomposition $\hat\gamma_{ij} = \sum_{a=1}^m \hat\nu_a (\hat\psi_i^a)^* \hat\psi_j^a$ and define an active space $\mathcal{A} = \mathrm{span}\{\hat\psi^a \mid \hat\nu_a > \varepsilon/4n\}.$ Then, for $\ket{v} \in \mathcal{A}^\perp$ we have $\langle v|\gamma|v\rangle \leq \langle v|\hat\gamma |v\rangle + \|\hat\gamma-\gamma\| = \varepsilon/3n$ and Lemma~\ref{lem:active} gives $\OPT \leq \OPT_{\mathcal{A}} +\varepsilon/3$. Moreover, $r \coloneqq \dim\mathcal{A} \leq 6n^2/\varepsilon = O(n^2/\varepsilon)$ as before. We assume an $\varepsilon$-covering net of $\mathcal{S}_{\mathcal{A}}$ is chosen as in Theorem~\ref{thm:upper}, and we refer there for further details. Our strategy will be to apply agnostic quantum threshold search to the family of Slater projectors in this covering net. \par 
    In particular, we choose a net whose constituents have radius $\varepsilon / 6n$, which covers the manifold $\mathcal{S}_{\mathcal{A}}$ with fidelity accuracy $\varepsilon/3$. We denote this as $\{\ket{S_1},\ldots, \ket{S_M}\}$ with 
    \begin{equation}
        M \leq \left(\frac{6nc}{\varepsilon}\right)^{2n(r-n)}
    \end{equation}
    where $c$ is the same constant appearing in Eq.~\eqref{eq:2nc}. For each $k$, let $A_k = \ket{S_k}\!\bra{S_k}$. We apply Lemma~\ref{lem:agnosticQTS} to $A_1,\ldots, A_M$ with accuracy $\varepsilon/3$. In the event of no failures, this procedure returns $i_\star$ where $\ket{S^\star} \coloneqq \ket{S_{i_\star}}$ satisfies 
    \begin{equation}
        \langle S^\star |\rho|S^\star \rangle \geq \max_k\langle S_k|\rho|S_k\rangle - \varepsilon/3.
    \end{equation}
    By design of the net, we have $\max_k\langle S_k|\rho|S_k\rangle\geq \OPT_{\mathcal{A}}-\varepsilon/3$. Moreover, by our choice of $\mathcal{A}$ we have $\OPT_{\mathcal{A}}\geq \OPT -\varepsilon/3$. Combining these inequalities, we have
    \begin{equation}
         \langle S^\star |\rho|S^\star\rangle \geq \OPT -\varepsilon.
    \end{equation}
    Finally, combining $r=O(n^2/\varepsilon)$ with $\log M = O\left(n(r-n)\log(n/\varepsilon)\right)$ and substituting into $n_{\mathrm{AQTS}}$ (Eq.~\eqref{eq:naqts}) yields $ N=  \widetilde O\left(n^6/\varepsilon^4\right)$, as was to be shown. 
\end{proof}

In Theorem~\ref{thm:quantum}, we have assumed the estimator $\hat\gamma$ was given. In practice, $\hat\gamma$ must be constructed using tomography on $\rho$, and we must count the number of copies of $\rho$ needed for this step of the procedure. We account for this using Lemma~\ref{lem:noisy-active}, which establishes using fermionic classical shadows~\cite{zhao2021fermionic,wan2023matchgate,christensen2026learning} that an estimator $\hat \gamma$ for the 1-RDM $\gamma$ satisfying $\|\hat\gamma-\gamma\|\leq \eta$ can be found with probability at least $1-\delta$ using $n_{\mathrm{fcs}} = O\left(\frac{m^2\log(m/\delta)}{\eta^2}\right)$ quantum copies of $\rho$. In the context of the previous theorem, we need $\eta \sim \varepsilon/n$, and hence the sample complexity from this estimation stage is $O\left(\frac{m^2n^2\log(m/\delta)}{\varepsilon^2}\right) = \mathrm{poly}(m,n,1/\varepsilon)$. \par
At this stage, we have provided algorithms for learning the closest Slater in both the classical and quantum access models. These correspond to the upper bounds in Table~\ref{tab:results}. Next, we turn to computational hardness lower bounds, which limit the possible efficiency of any algorithm for this problem.

\begin{figure*}
    \centering
\includegraphics[width=\linewidth]{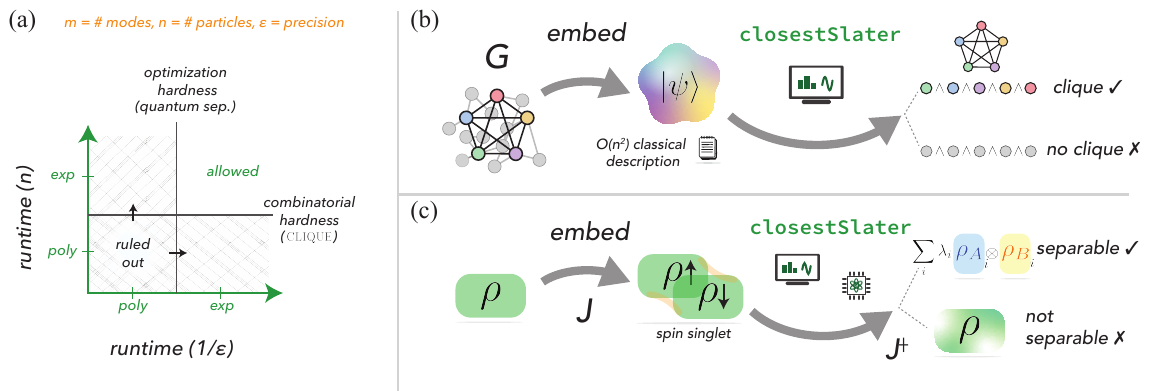}
    \caption{\textbf{Overview of computational hardness.} (a) Learning the closest Slater determinant is possible in time polynomial in $m$ and exponential in $\poly(n,1/\varepsilon)$. Substantially improving the dependence on $n$ or $1/\varepsilon$ is ruled out by standard computational complexity conjectures. (b) To prove hardness for the $n$ axis, we embed a graph instance $G$ of the $n$-\textsc{Multicolored Clique} problem into a classical description of a fermionic state, which could be efficiently solved by an efficient closest Slater algorithm. (c) To prove hardness along the $1/\varepsilon$ axis, a bipartite state $\rho$ is embedded by $J$ into the opposite-spin sector of a fermionic state in such a way that a classical or quantum algorithm for the closest Slater would decide whether $\rho$ is separable, solving the quantum separability problem~\cite{gharibian2010strong}.}
    \label{fig:comphardness}
\end{figure*}

\section{Computational hardness}\label{sec:hardness}

So far, we have provided algorithms for learning the closest Slater determinant whose time complexities are polynomial in $m$, the single-particle Hilbert space dimension, at constant particle number $n$ and additive error $\varepsilon$. However, they are exponential in $\poly(n,1/\varepsilon)$ at fixed $m$. In this section, we establish computational hardness lower bounds in both $n$ and $1/\varepsilon$ based on standard complexity-theoretic assumptions, hence showing that this time complexity cannot be substantially improved along either axis, albeit with one substantive problem left open: the possibility of subexponential dependence in $n$ for a quantum algorithm. In the classical case, the following results render our upper bound essentially time-optimal. Figure~\ref{fig:comphardness} summarizes the two hardness reductions and the resulting complexity landscape.\par 

\subsection{Classical hardness in particle number}

First, we prove a hardness lower bound in $n$ for the classical setting. To prove this, we will use a standard hardness assumption from parameterized complexity, namely the hardness of the $n$-\textsc{Multicolored Clique} problem~\cite{CyganEtAl2015}. In this problem, one is given a
graph whose vertices are partitioned into $n$ colors
\begin{equation}
  V_1,\ldots,V_n,
\end{equation}
with each subset of size $p$, for a total of $np$ vertices. A \textsc{Multicolored Clique} is a choice of one vertex per color such that the $n$ chosen vertices are pairwise connected. Under the standard exponential time hypothesis (ETH), one does not expect an
algorithm running in time $f(n)p^{o(n)}$ for this problem, for any
computable $f$.
\par 
For our purposes, it is convenient to use a variation where the graph is guaranteed to have zero or one multicolored clique. This is without any significant loss of generality due to the Valiant-Vazirani theorem~\cite{ValiantVazirani1986}. One detail of appealing to this theorem is that we will phrase our assumption in the more general setting of randomized algorithms. Our self-contained assumption, therefore, will be as follows.

\begin{assumption}[$n$-\textsc{Multicolored Clique~\cite{CyganEtAl2015,ValiantVazirani1986}}]
\label{ass:rETH}
Let $G$ be an $n$-colored graph, with each color class of size $p$. Given $G$, there is no randomized algorithm running in time $f(n)p^{o(n)}$, for any computable function $f$, which distinguishes the following two cases (with success probability bounded above $1/2$ by a constant) under the promise that one is true:
\begin{enumerate}
  \item $G$ has no $n$-multicolored clique
  \item $G$ has exactly one $n$-multicolored clique.
\end{enumerate}
\end{assumption}

\noindent Under this assumption, we prove the following:

\begin{restatable}[Classical complexity lower bound]{theorem}{classicalcomplexitythm}
\label{thm:lower}
Fix any $\lambda<1$. Under Assumption~\ref{ass:rETH}, there exists a
constant $\varepsilon_\lambda>0$ such that no randomized algorithm running in time $f(n)m^{o(n)}$ can solve the following task in general: given a succinct classical
description of a nonzero $n$-fermion vector $|\Phi\rangle$ on $m$ modes,
output a Slater $| S\rangle$ satisfying
\begin{equation}\label{eq:Scontradiction}
  \frac{|\langle S|\Phi\rangle|^2}{\|\Phi\|^2}
  \ge
  \OPT-\varepsilon_\lambda,
\end{equation}
where $\OPT = \max_{\ket{S}\in \mathcal{S}} |\langle S|\Phi\rangle|^2/\|\Phi\|^2$, 
with success probability at least $0.99$, even under the promise that $\OPT \ge \lambda.$ Here a succinct classical description means a polynomial-size circuit
which, on inputting a set $I\subset [m]$ with corresponding Slater determinant $\ket{I}$, returns the (possibly unnormalized) amplitude $\langle I|\Phi\rangle$.
\end{restatable}

Below we outline the main ideas of the proof, leaving the full argument to Appendix~\ref{sec:classicalhardness}. We embed the $n$-\textsc{Multicolored Clique} problem inside a fermionic state, as shown schematically in Fig.~\ref{fig:comphardness}b. Let $G$ be an $n$-colored graph satisfying the promise of Assumption~\ref{ass:rETH}. It will be slightly more convenient to have the promise of either one or two $n$-multicolored cliques, so we add to $G$ a disjoint clique $D$ and then assign one fermionic mode to each vertex. The fermionic state is chosen as $\ket{\Phi_G} = \ket{D}+\alpha \ket{C}$, with $\alpha>1$. Here, $\ket{C}$ is the Slater for the hidden clique if it exists and $\ket{C} = 0$ otherwise. The state $\ket{\Phi_G}$ has a succinct description, since each amplitude $\langle I |\Phi_G\rangle$ equals $1,\alpha,$ or $0$ and can be computed from the graph in $O(n^2)$ time. Because $D$ and $C$ share no orbitals, the closest Slater is $\ket{C}$ when it exists and $\ket{D}$ otherwise (Lemma~\ref{lem:Superposition}). Taking a single overlap with $\ket{D}$ then reveals the solution to $n$-\textsc{Multicolored Clique}. Hence, an efficient closest Slater algorithm would violate Assumption~\ref{ass:rETH}. \par 

This establishes that the classical time complexity cannot be substantially improved beyond the one presented in Sec.~\ref{sec:algorithm}, and in particular no subexponential-in-$n$ algorithm should be expected. We note that this argument is based on a classical access model. Our quantum algorithm also has exponential-in-$n$ runtime, but the argument above does not provide a lower bound in the quantum setting. \par

\subsection{Hardness in precision}

A second axis of computational hardness is the scaling with respect to the precision $\varepsilon$. Here, we show that the exponential $1/\varepsilon$ scaling in the time complexity of the classical and quantum algorithms cannot be improved to polynomial when we allow the input state to be mixed, based on widely-believed complexity-theoretic assumptions. In particular, we prove the following theorem, shown schematically in Fig.~\ref{fig:comphardness}c. 

\begin{restatable}[Computational hardness]{theorem}{comphardthm}\label{thm:comp-hard}
    There exists a positive constant $c$ such that if there exists a $\poly(m)$-time quantum algorithm that takes as input the classical description of any (possibly mixed) density matrix $\rho$ of $n=2$ fermions on $m$ modes and outputs a Slater determinant $\ket{S}\in \mathcal{S}$ satisfying
    \begin{equation}
    \langle S|\rho|S\rangle \geq \mathsf{OPT} - 1/m^c
    \end{equation}
    with probability at least $2/3$, then $\mathsf{NP}\subseteq \mathsf{BQP}$.
    Moreover, if such a randomized classical algorithm exists, then $\mathsf{NP}\subseteq \mathsf{BPP}$.
\end{restatable}

In particular, this proved absence of $\poly(m)$-time algorithms implies the absence of $\poly(1/\varepsilon)$ algorithms, by choosing $\varepsilon=1/m^c$. 
Hence, we have the $(1/\varepsilon)^{\omega(1)}$ runtime lower bound quoted in Table~\ref{tab:results}.

We outline the main idea of the proof, leaving the full argument to Appendix~\ref{ref:hardnessepsilon}. The basic idea is to reduce the problem to the quantum separability problem, which asks whether a bipartite quantum state $\rho$ is entangled or separable. When $\rho$ is located within an inverse polynomial (with respect to dimension) distance from the border of the set of separable quantum states, this problem was shown to be $\mathsf{NP}$-hard~\cite{gharibian2010strong}. To relate this problem with learning Slater determinants, we may regard two fermions in an even number of modes $m$ as two fermions in $m/2$ modes with opposite spins, and the Slater fidelity of the former corresponds to the separability of the latter when the spin degree of freedom is anti-symmetrized. In this case, the ``spatial'' wavefunction of the fermions in a Slater determinant state is in a product state. This gives a reduction between testing separability and the problem of learning the closest Slater determinant. This hardness applies to both the classical and quantum access models; indeed, even if we have access to the \textit{complete} classical description of the density matrix, solving the closest Slater problem will take superpolynomial time in $1/\varepsilon$ unless $\mathsf{NP}\subseteq \mathsf{BPP} / \mathsf{BQP}$, which is widely believed to be false.

\section{Optimization landscape}\label{sec:highfid}

So far, we have assumed that the fermionic state $\ket\psi$ is arbitrary. In this section, we study the situation where $\ket\psi$ is already known to be close to a Slater determinant. Loosely speaking, this corresponds to $\ket\psi$ being weakly correlated, and hence still captures many interesting quantum systems including filled-valence molecules and weakly interacting phases of matter. Our specific result is the following. Suppose that we have a trial Slater $\ket{G}$ satisfying $|\langle G|\psi\rangle|^2  > 2/3$; then $\ket{G}$ is close to the \textit{optimal} Slater determinant, in a manner to be made precise. This guarantee is useful for a non-convex optimization problem and is noteworthy for its independence of $m,n$, or model details.
\par 
Practically, this result could prove useful in the following scenario. Suppose we have a heuristic algorithm to find the closest Slater, such as gradient ascent or Hartree-Fock-style iteration. If we arrive at a Slater determinant $\ket{G}$ with fidelity $\mathrm{F}>2/3$, then it suffices to perform some local optimization to refine the ans\"atz $\ket{G}$ into the globally closest Slater determinant. In particular, if $\ket{G}$ satisfies a further stationarity condition, to be defined, then $\ket{G}$ is globally optimal. Since the original problem is a non-convex optimization, for which a local optimum need not be a global one, this guarantee can prove valuable. We will also prove that this condition cannot be improved, in the sense that the constant $2/3$ is optimal in general.\par

To proceed, let $|G\rangle$ be an arbitrary Slater and choose an orthonormal basis $g_1, \ldots, g_n$ such that $\ket{G} = c_{g_1}^\dagger\cdots c_{g_n}^\dagger|0\rangle$. Let $f_1,\ldots, f_{m-n}$ be an orthonormal basis for the unoccupied orbitals and define single-particle-hole excitations
\begin{equation}
    \ket{G_{a\mu}} \coloneqq c_{f_a}^\dagger c_{g_\mu}\ket{G}. 
\end{equation}
Let $\ket{\psi}$ be our fermionic state. We define $\alpha = |\langle G|\psi\rangle|$ and 
\begin{equation}
    \kappa = \alpha - \sqrt{2}\sqrt{1-\alpha^2},
\end{equation}
noting that $\alpha^2>2/3 \iff \kappa > 0$. Finally, we define 
\begin{equation}
    B_G \coloneqq\sqrt{\sum_{a\mu}|\langle G_{a\mu}|\psi\rangle|^2}.
\end{equation}
By Lemma~\ref{lem:brillouin}, $B_G=0$ is precisely the condition that fidelity is stationary with respect to $\ket{G}$. We refer to $B_G$ as the \textit{Brillouin residual}~\cite{Brillouin1933}.
\begin{theorem}[$2/3$ certificate] \label{thm:certificate}
Let $\ket{\psi}$ be a normalized fermionic state and let $\ket{G}$ be a Slater determinant. Define $\kappa$ and $B_G$ as above. Suppose
\begin{equation}
    |\langle G|\psi\rangle |^2 >2/3
\end{equation}
and let $\OPT = \max_{\ket{S}\in\mathcal{S}} |\langle S|\psi\rangle |^2$. Then
\begin{equation}
    \OPT \leq \left(|\langle G|\psi\rangle| + \frac{B_G^2}{2\kappa}\right)^2. 
\end{equation}
In particular, if $B_G = 0$ then $\ket{G}$ is the unique globally closest Slater determinant to $\ket\psi$, up to phase.  
\end{theorem}

We note that $2/3$ is a worst-case threshold, and is far from tight for
typical states. Across various ensembles of random and physical quantum states, we have numerically observed that spurious stationary points only occur at fidelities well below
$2/3$ (Fig.~\ref{fig:numerics}c and Sec.~\ref{sec:numerics}). As a
state-independent guarantee, however, the constant $2/3$ cannot be improved, as
we will show.
\begin{proof}[Proof of Theorem~\ref{thm:certificate}]
Let $\ket{H}$ be an arbitrary Slater determinant. We will compare $|\langle H|\psi\rangle|$ to $|\langle G|\psi\rangle|$. Following the notation of Lemma~\ref{lem:tail}, we decompose $\ket{H}$ into particle-hole sectors relative to $\ket{G}$:
\begin{equation}
  \ket{H}
  =  \ket{h_0}+\ket{h_1}+\ket{h_{\ge 2}}.
\end{equation}
Here $\ket{h_0}$ is proportional to $\ket{G}$, $\ket{h_1}$ lies in the single-particle-hole sector spanned by the $\ket{G_{a\mu}}$, and $\ket{h_{\geq 2}}$ is the remainder. We choose
the phase of $\ket{H}$ so that $\ket{h_0}=s\ket{G}$ with $s=|\langle H|G\rangle|$, and we'll write $p_1:=\norm{\ket{h_1}}^2,  q^2:=\norm{\ket{h_{\ge 2}}}^2.$ We perform the analogous particle-hole decomposition of
the target state $\ket\psi$ relative to $\ket G$, writing
\begin{equation}
  \ket{\psi}
  =
  \alpha\ket{G}
  +
  \ket{\psi_1}
  +
  \ket{\psi_{\ge 2}},
\end{equation}
where we've absorbed a possible phase into $\ket{G}$ and we note $\|\ket{\psi_1}\| = B_G$. We also note $\|\ket{\psi_{\geq 2}}\|\leq \sqrt{1-\alpha^2}$. Using these decompositions and applying the Cauchy-Schwarz inequality to each sector, we find
\begin{equation}
    |\langle H|\psi\rangle| \leq \alpha s+ B_G \sqrt{p_1} + \sqrt{1-\alpha^2} q. 
\end{equation}
We now make use of two bounds. First, $ p_1
  \le
  1-s^2
  \le
  2(1-s).$ Second, by Lemma~\ref{lem:tail}, $q
  \le
  \sqrt{2}\,(1-s).$ Then 
\begin{align}
  |\langle H|\psi\rangle|
  &\le
  \alpha s
  +
  B_G\sqrt{2(1-s)}
  +
  \sqrt{2}\sqrt{1-\alpha^2}(1-s)
  \nonumber\\
  &=
  \alpha-\kappa(1-s)+B_G\sqrt{2(1-s)}.
\end{align}
Since $\kappa>0$, this expression is uniformly bounded over
$0\leq s \leq 1$ by
\begin{equation}
  \max_{s\in [0,1]}
  \left\{\alpha
    -\kappa(1-s)+B_G\sqrt{2(1-s)}
  \right\}
  \leq\alpha+
  \frac{B_G^2}{2\kappa}.
\end{equation}
Therefore every Slater determinant $\ket{H}$ satisfies $ |\langle H|\psi\rangle|
  \le
  \alpha+\frac{B_G^2}{2\kappa}.$ Taking the maximum over $\ket{H}$ and squaring proves the claim. For uniqueness, we note that when $B_G=0$ the prior bound indicates $|\langle H|\psi\rangle|\leq \alpha - \kappa(1-s)$ with $\kappa>0$ and $s\leq 1$. Since $\OPT=\alpha^2$, any global maximum $\ket{H}$ attains $|\langle H|\psi\rangle|=\alpha$, forcing $\kappa(1-s)\leq 0$ and hence $s=1$. Thus $\ket{H}=\ket{G}$ up to phase. 
\end{proof}

\begin{corollary}[High-fidelity stationary Slaters are optimal]
\label{cor:local}
Let $\ket{G}$ be a Slater determinant with $|\langle G|\psi\rangle|^2>2/3.$ If $\ket G$ is stationary, meaning $  B_G=0,$ then $\ket G$ is the unique globally closest to $\ket\psi$, up to phase.
\end{corollary}

The above theorem motivates the following question: can the threshold of $2/3$ fidelity be improved? That is, is there a lower number for which the theorem and corollary hold? We answer this in the negative in the theorem below, establishing $2/3$ as the sharp threshold above which the vanishing of $B_G$ implies global optimality for the closest Slater determinant problem. \par

\begin{theorem}[Sharpness of the $2/3$ threshold]
\label{thm:sharp}

\begin{figure*}
    \centering    \includegraphics[width=\linewidth]{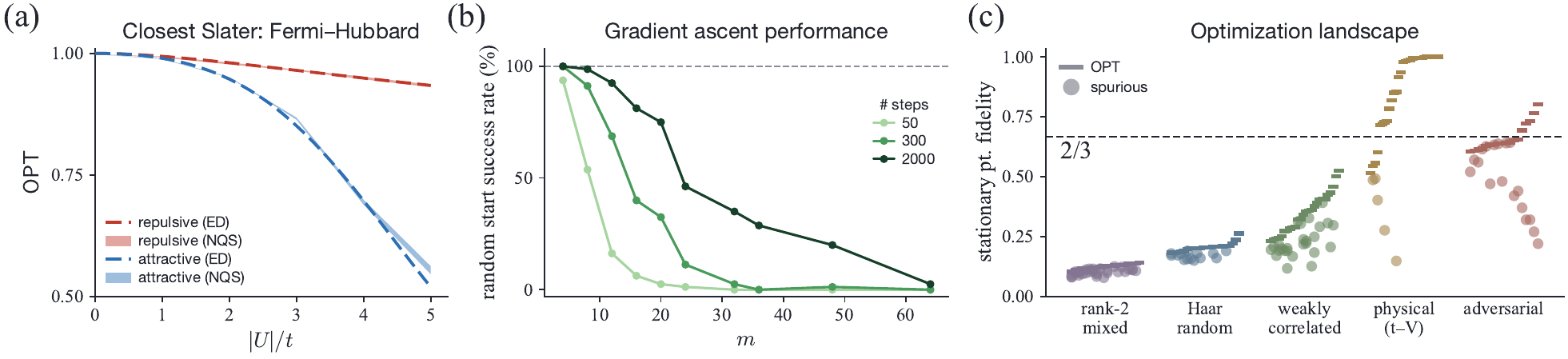}    \caption{\textbf{Closest Slater implementation and optimization
landscape. } (a) Fidelity $\OPT = \max_{\ket{S}} |\langle S|\psi\rangle|^2$ of
the $4\times4$ Fermi--Hubbard ground state (cylinder, $n_\uparrow=n_\downarrow=2$) to the closest Slater versus $|U|/t$, for both repulsive and attractive  interactions. Dashed
lines are exact diagonalization (ED) and shaded bands are from trained neural quantum states (NQS), with width the Monte Carlo sampling standard deviation. (b) Fraction of randomly initialized optimizations, for a simple gradient ascent optimizer, that reach the global
$\OPT$, versus system size $m = 2L_xL_y$  ($n_\uparrow=n_\downarrow=2,
U/t = 2$). At any fixed $\#$ steps, as $m$ increases gradient ascent stalls and the success rate collapses, whereas our algorithm is guaranteed to reach $\OPT$ (dashed line). (c) Fidelities of spurious
stationary Slater determinants versus $\OPT$ for five fermionic state
ensembles ($n=5$, $m=10$): random rank-two mixed states, Haar-random pure states, weakly correlated (i.e. few-Slater superpositions), ground states of $t-V$ models, and the engineered adversarial states (c.f. Thm.~\ref{thm:sharp}). States ordered by increasing $\OPT$ within each ensemble.  Details in App.~\ref{app:numerics}. No spurious stationary point fidelity exceeds $2/3$, as guaranteed by Thm.~\ref{thm:certificate}, with most well below. }    \label{fig:numerics}
\end{figure*}

For every $c<2/3$, there exist integers $m,n$, a normalized $n$-fermion
state $\ket{\psi}$, and a Slater determinant $\ket{G}$ such that $|\langle G|\psi\rangle|^2>c$ and $B_{G}=0$ but $\ket{G}$ is not globally closest to $\ket{\psi}$. Hence no $m,n$-independent statement of the form of Theorem \ref{thm:certificate} can hold for $c<2/3$.
\end{theorem}
\begin{proof}
    We'll construct an explicit counterexample. Take $m=2n$ fermionic modes with orthonormal orbitals $e_1,\ldots, e_n, f_1,\ldots, f_n$ and let $\ket{G} = c_{e_1}^\dagger \cdots c_{e_n}^\dagger \ket{0}$. For $t>0$, define the rotated Slater determinant
\begin{equation}
     \ket{H_t} = \prod_{i=1}^n\frac{c_{e_i}^\dagger + tc_{f_i}^\dagger}{\sqrt{1+t^2}}\ket{0}.   
\end{equation}
We'll decompose $\ket{H_t}$ into particle-hole sectors relative to $\ket{G}$: 
\begin{equation}
    \ket{H_t} = s_t\ket{G} + \ket{h_{t,1}} + \ket{h_{t,\geq 2}}.
\end{equation}
From explicit calculation, we have $s_t = (1+t^2)^{-n/2}$ and 
\begin{equation}
     q_t \coloneqq \|\ket{h_{t,\geq 2}}\| = [1-(1+nt^2)s_t^2 ]^{1/2}.
\end{equation}
We also note that $\lim_{t\to 0} \frac{q_t}{1-s_t} = \sqrt{2(n-1)/n}$ which approaches $\sqrt{2}$ as $n\to\infty$. Now choose $\alpha$ such that $ c < \alpha^2 < 2/3.$ Since $\alpha^2 < 2/3$, we have $ \frac{\alpha}{\sqrt{1-\alpha^2}} < \sqrt{2}.$ We choose $n$ large enough and $t>0$ small enough that $ \frac{q_t}{1-s_t} >\frac{\alpha}{\sqrt{1-\alpha^2}}$ and define the state
\begin{equation}
    \ket{\psi} = \alpha \ket{G} + \frac{\sqrt{1-\alpha^2}}{q_t}\ket{h_{t,\geq 2}},
\end{equation}
which is properly normalized, $\langle \psi|\psi\rangle =1$. Moreover, $|\langle G|\psi\rangle|^2 = \alpha^2 > c$ and $B_G = 0$, since $\ket{\psi}$ has no single-particle-hole component relative to $\ket{G}$. However, $\ket{G}$ is not the globally closest Slater to $\ket{\psi}$; rather, we have
\begin{equation}
    \begin{aligned}
        |\langle H_t|\psi\rangle|^2 &=(\alpha s_t + \sqrt{1-\alpha^2} q_t)^2\\
        &> (\alpha s_t + \alpha(1-s_t))^2= \alpha^2
    \end{aligned}
\end{equation}
so $\ket{H_t}$ is closer. This construction provides the desired counterexample. 
\end{proof}

We will refer to the states $\ket{\psi}$ defined above as ``adversarial" states, as they are constructed to be adversarially difficult for a closest-Slater algorithm. Together with Theorem~\ref{thm:certificate}, the theorem above establishes F$=2/3$ as a sharp threshold in the optimization landscape. Above it, every stationary point is the global maximum, while below it, stationary points need not be globally optimal. We have stated these high-fidelity results for pure states for clarity of exposition, but these results extend to the mixed-state setting, with the residual $B_G$ replaced by a suitable mixed-state analog. Moreover, the algorithmic and hardness results carry over as well. We prove these extensions in Appendix~\ref{app:mixed}.

\section{Numerical results}
\label{sec:numerics}

In this section, we apply the closest Slater classical algorithm to ground states of interacting fermionic models and characterize the optimization landscape. We summarize results here and leave the details of the numerics to Appendix~\ref{app:numerics}. First, we study the Fermi--Hubbard model 
\begin{equation}
    H \;=\; -t\sum_{\langle ij\rangle,\sigma} c^\dagger_{i\sigma}c_{j\sigma} \;+\; U\sum_i n_{i\uparrow}n_{i\downarrow}
\end{equation}
on an $L_x\times L_y$ square lattice with  $(n_\uparrow,n_\downarrow)$ particles of each spin. At each $U/t$, we solve for the ground state using exact diagonalization or neural quantum states. We then evaluate $\OPT$ using the classical algorithm described in Sec.~\ref{sec:algorithm}, which involves computing the 1-RDM, restricting to an active space, and searching over a covering net of the corresponding manifold of Slaters. Fig.~\ref{fig:1}c shows the closest Slater fidelity $\OPT$ of the ground state as a function of $|U|/t$ using exact diagonalization. Fig.~\ref{fig:numerics}c shows the same result (reproduced) as well as the result using a variationally optimized neural quantum state.
\par 
The value of $\OPT$ is unity in the noninteracting limit and decreases as $|U|$ grows and correlations build. $\OPT$ decreases more rapidly for attractive interactions, $U<0$, than for repulsive interactions $U>0$. This is a manifestation of the pairing correlations of the ground state of the attractive Fermi-Hubbard model, which preclude high-fidelity representation by a Slater determinant. \par 

Our algorithm requires only sample access, as demonstrated in
Fig.~\ref{fig:numerics}a where $\OPT$ is calculated using Monte Carlo sampling of a neural quantum state (NQS). The NQS representation~\cite{carleo2017solving} provides a versatile ans\"atz for strongly correlated electron wavefunctions and allows Monte Carlo estimation of sparse observables. In practice, NQS have proven competitive across condensed matter, quantum chemistry, and nuclear physics~\cite{luo2019backflow,cassella2023discovering,chen2024empowering,rende2024simple,zhao2024empirical,teng2025solving,geier2025self,pfau2020ab,hermann2020deep,pfau2024accurate,yang2023deep,adams2021variational,gnech2024distilling,fore2025investigating}. Motivated by this broad applicability, here we demonstrate that computing $\OPT$ by sampling a NQS can be straightforward and informative. We choose a linear combination of Slater-Jastrow terms as the NQS ans\"atz, and results are close to ED with the deviation originating from imperfect training and the width originating from Monte Carlo standard deviation. \par

In Fig.~\ref{fig:numerics}b we use $\OPT$ across various system sizes $m=2L_xL_y$ as the ground truth against
which to benchmark heuristic optimization. We run a simple implementation of local gradient ascent on the
Slater manifold from random initializations and record the fraction of starts
that reach the global $\OPT$. This fraction falls sharply as the lattice grows for any fixed number of iteration steps,
from near unity at the smallest sizes to near zero at $m=64$. This gives evidence that a na\"ive gradient ascent for solving this optimization problem may only be reliable at small system sizes. In contrast, the classical algorithm we propose comes with guarantees and has $\mathrm{poly}(m)$ scaling for fixed number of particles and precision.  
\par 

In Fig.~\ref{fig:numerics}c, we capture aspects of the optimization landscape for the closest Slater problem in various ensembles of fermionic states. In particular, we study rank-2 random mixed states, Haar-random pure states, weakly correlated states (comprising superpositions of three to six Slaters with comparable weights), ground states of spinless disordered $t-V$ chains, and the adversarial states constructed in the proof of Theorem~\ref{thm:sharp}. For each state we plot its value of $\OPT$ and the fidelities of spurious stationary points attained as optimization endpoints of randomly initialized simple local gradient ascent. We note that the results shown are representative within each ensemble. Across every ensemble, all spurious stationary points lie strictly below $2/3$, consistent with Theorem~\ref{thm:certificate} and its mixed-state extension. 

The optimization landscapes have several noteworthy features. The spurious stationary points lie well below $2/3$ for the non-adversarial ensembles: those found are at most $\approx 0.13, 0.19, 0.39$, and $0.49$ for mixed, Haar, weakly correlated, and physical. The latter is due to a translation symmetry-broken charge density wave ground state at large repulsive interaction, which at finite sizes results in the ground state assuming the form of an approximate superposition of Slaters (an exact superposition in the absence of disorder), and hence two closely competing closest Slaters. On the whole, we found it rare for $t-V$ model ground states to exhibit spurious stationary points. This suggests that the optimization landscape for physical states may be benign in general, though more comprehensive study is needed. Meanwhile, the random and weakly correlated state ensembles exhibit rugged landscapes with many spurious stationary points; however, these are less relevant to practical applications and the spurious stationary points concentrate at $\lesssim 0.25$ at this system size ($n=5,m=10$). The adversarial states, by contrast,
drive spurious stationary points up to the threshold ($0.642 \approx 2/3$ at
$n=16$). While confirming Theorem~\ref{thm:sharp}, this suggests that $2/3$ may be a worst case threshold.

\section{Discussion}
\label{sec:discussion}

We have studied the task of finding the closest Slater determinant to a given $n$-fermion state on $m$ modes, with closeness measured by fidelity. On the classical side, we have provided an algorithm (Theorem~\ref{thm:upper}) achieving this within $\varepsilon$ additive error with time complexity $m^{\text{poly}(n,1/\varepsilon)}$. Moreover, Theorem~\ref{thm:lower} establishes a matching parameterized hardness lower bound under the exponential time hypothesis, so the exponential-in-$n$ dependence is essentially unavoidable classically. On the quantum side, Theorem~\ref{thm:quantum} provides an algorithm with sample complexity to $\text{poly}(m,n,1/\varepsilon)$ via the quantum threshold search algorithm of Ref.~\cite{buadescu2021improved}. In both the classical and quantum settings, we have shown that the time complexity scaling with precision cannot be substantially improved. We have further established that the optimization landscape associated to this problem exhibits an interesting transition at a fidelity of $2/3$. In particular, Theorem~\ref{thm:certificate} certifies that any locally-optimal Slater with fidelity above $2/3$ is the globally optimal one. \par 

Our problem is an instance of the broader learning task of agnostic tomography~\cite{grewal2026agnostic,chen2025stabilizer} and is adjacent to several other optimization problems. It is similar to the Hartree-Fock approximation, which minimizes the energy of a fermionic Hamiltonian over Slater determinants and is the standard mean-field starting point in chemistry and condensed matter~\cite{RevModPhys.35.496,echenique2007mathematical}, but here a given state $\rho$ is the central object rather than a Hamiltonian. Moreover, a close bosonic analog to our problem is the closest product state learning problem studied by Bakshi et al.~\cite{bakshi2025learning}. We note that Bakshi et al. focus on a quantum access model obtain more favorable complexities than found here. Moreover, in the special case $\OPT = 1$, when the state is exactly a Slater determinant, the algorithms of Aaronson and Grewal~\cite{aaronson_et_al:LIPIcs.TQC.2023.12}, O'Gorman~\cite{o2022fermionic}, and Christensen and Zhao~\cite{christensen2026learning} achieve $\text{poly}(m,n,1/\varepsilon)$ runtime using the 1-RDM alone, though these algorithms do not generally extend to $\OPT < 1$. \par

There are several interesting directions for future study. First, it would be interesting to relax particle number conservation and study the problem of learning the closest fermionic Gaussian state \cite{mele2025efficient}, which could better accommodate phenomena such as superconductivity; we note, however, that superconducting order can be straightforwardly probed at fixed particle number via off-diagonal long-range order. Second, our quantum protocol attains
$\poly(n,1/\varepsilon)$ sample complexity but retains an exponential-in-$n$
runtime, and it remains open whether the quantum access model admits a
subexponential-time algorithm or, conversely, a runtime lower bound on $n$ mirroring
the classical case. Third, it would be interesting to apply our algorithm at reasonable $n$ to strongly correlated systems to detect the possibility of ``free fermions in disguise", or unexpectedly high fidelity with the Slater manifold. Finally, a number of natural extensions suggest themselves, such as generalizing to the presence of symmetry-averaged Slater determinants, which may have $\OPT \ll 1$ but are still qualitatively Slater-like, or other physically motivated metrics besides fidelity.

\par 
\textit{Note added.---} Upon completing this work, we became aware of independent and concurrent work also addressing the problem of learning the closest Slater determinant~\cite{ewintang}. The authors report similar results including a high-fidelity algorithm and a low-fidelity algorithm with a superpolynomial dependence in the number of particles. 

\par 
\textit{Acknowledgments. } We thank Jason Alicea, Patrick Ledwith, Gal Shavit, and Andrew Zhao for helpful comments. NP and HZ acknowledge Anthropic's Claude for assistance with literature search, coding, and technical lemmas. NP acknowledges support
from the Walter Burke Institute for Theoretical Physics at Caltech.
HZ acknowledges support from the Institute for Quantum Information and Matter, an NSF Physics Frontiers Center (PHY-2317110).
DDD was supported by the National Science Foundation Graduate Research Fellowship under Grant No. 2141064.
This work is supported by the National Science Foundation under Cooperative Agreement PHY-2019786 (The NSF AI Institute for Artificial Intelligence and Fundamental Interactions, http://iaifi.org/).

\bibliographystyle{apsrev4-2}
\bibliography{bib}
\newpage
\clearpage
\appendix

\section{Computational hardness}
\label{sec:comphard}

\subsection{Classical hardness in $n$ from parameterized complexity}
\label{sec:classicalhardness}

\classicalcomplexitythm*

\begin{proof}[Proof of Theorem~\ref{thm:lower}]
We prove this by embedding the $n$-\textsc{Multicolored Clique} problem into the problem of learning the closest Slater. Assume we have a graph with color classes $V_1,\ldots, V_n$, each of size $p$, as in Fig.~\ref{fig:mcclique}, and we are
promised that there is zero or one $n$-multicolored clique. Under Assumption~\ref{ass:rETH}, there is no randomized
algorithm running in time $f(n)p^{o(n)}$ which distinguishes these two
cases.
\par 
We first augment the graph by adding one extra vertex $d_i$ to each
color class $V_i$. We add edges pairwise between all the $d_i$, but no edges connecting the $d_i$ to the original vertices. Hence the $d_i$ form a disjoint $n$-multicolored clique
\begin{equation}
  D=(d_1,\ldots,d_n).
\end{equation}
After this augmentation, the graph is promised to contain either only the disjoint
clique $D$, or the disjoint clique $D$ together with exactly one other $n$-multicolored clique $C$. We call this augmented graph $G$.

Next, we encode this graph into a fermionic state as follows. We set the particle number
to be $n$ and assign one fermionic mode to each vertex of the graph. Thus the
total number of modes is
\begin{equation}
 m= |G|=n(p+1).
\end{equation}
Next, let $\mathcal{K}(G)$ denote the set of $n$-multicolored cliques distinct from $D$ in the graph, so $|\mathcal{K}(G)| = 0$ or $1$. Each subgraph $A\subset G$ of size $|A|=n$ naturally defines an $n$-particle Slater determinant, which we will denote $\ket{A}$. Using this convention, we define an unnormalized fermionic state
\begin{equation}
    \ket{\Phi_G} = \ket{D} + \alpha \sum_{C\in \mathcal{K}(G)} \ket{C},
\end{equation}
where $\alpha>1$ is chosen so that $\frac{\alpha^2}{1+\alpha^2} > \lambda$. We depict $\ket{\Phi_G}$ heuristically in Fig.~\ref{fig:encoding}.  \par 

\begin{figure}[t]
\centering
\begin{tikzpicture}[scale=0.5]
  \definecolor{pcA}{HTML}{F2A6A6} 
  \definecolor{pcB}{HTML}{A9C8EE} 
  \definecolor{pcC}{HTML}{A9DDBF} 
  \definecolor{pcD}{HTML}{CBB4E8} 
  \definecolor{pcE}{HTML}{F1D49B} 
  \definecolor{cliquecol}{HTML}{23252F}
  \tikzset{
    bg/.style    ={cliquecol!22, line width=0.7pt},
    cle/.style   ={cliquecol, line width=0.9pt, line cap=round},
    vtx/.style   ={circle, draw=cliquecol!35, line width=0.5pt,
                   minimum size=8pt, inner sep=0pt},
    vclq/.style  ={circle, draw=cliquecol, line width=1.2pt,
                   minimum size=10pt, inner sep=0pt},
  }

  \coordinate (v0)  at (0.000,3.400);
  \coordinate (v1)  at (1.383,3.106);
  \coordinate (v2)  at (2.527,2.275);
  \coordinate (v3)  at (3.234,1.051);
  \coordinate (v4)  at (3.381,-0.355);
  \coordinate (v5)  at (2.944,-1.700);
  \coordinate (v6)  at (1.998,-2.751);
  \coordinate (v7)  at (0.707,-3.326);
  \coordinate (v8)  at (-0.707,-3.326);
  \coordinate (v9)  at (-1.998,-2.751);
  \coordinate (v10) at (-2.944,-1.700);
  \coordinate (v11) at (-3.381,-0.355);
  \coordinate (v12) at (-3.234,1.051);
  \coordinate (v13) at (-2.527,2.275);
  \coordinate (v14) at (-1.383,3.106);

  \foreach \a/\b in {%
    0/3,0/4,0/6,0/7,0/13,1/13,2/3,2/4,2/8,2/11,3/4,3/5,3/6,3/7,3/8,3/10,%
    4/5,4/7,4/8,4/13,4/14,5/8,6/7,6/10,6/13,7/13,8/13,9/11}
    {\draw[bg] (v\a)--(v\b);}

  \foreach \a/\b in {1/5,1/9,1/12,1/14,5/9,5/12,5/14,9/12,9/14,12/14}
    {\draw[cle] (v\a)--(v\b);}

  \foreach \i/\c in {0/pcD,1/pcE,2/pcB,3/pcC,4/pcE,5/pcD,6/pcA,7/pcE,%
    8/pcA,9/pcC,10/pcC,11/pcD,12/pcA,13/pcB,14/pcB}
    {\node[vtx, fill=\c] at (v\i) {};}

  \foreach \i/\c in {1/pcE,5/pcD,9/pcC,12/pcA,14/pcB}
    {\node[vclq, fill=\c] at (v\i) {};}
\end{tikzpicture}
\caption{\textbf{Computational hardness mapping}. A graph $G$ whose vertices have $n=5$ different colors with $p=3$ vertices of each color. A multicolored clique is a choice of one vertex per color with all
$\binom{5}{2}$ pairs connected. Under the exponential time hypothesis (ETH), no algorithm can detect a multicolored clique in time $f(n)p^{o(n)}$, ruling out anything substantially faster than exhaustive search over the $p^n$ cliques. Exactly
one multicolored clique is present here, highlighted in bold. In
Theorem~\ref{thm:lower}, we encode $n$-\textsc{Multicolored Clique} into the problem of finding the closest Slater determinant to a suitably constructed fermionic state
$\ket{\Phi_G}$ (Fig.~\ref{fig:encoding}).}
\label{fig:mcclique}
\end{figure}
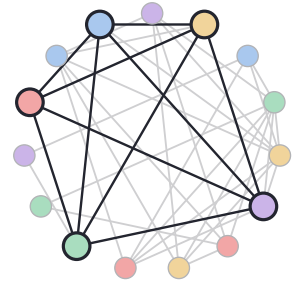

Next, we establish that $\ket{\Phi_G}$ indeed has a succinct classical description. To see this, we fix any set $I \subset [m]$ of size $|I|=n$ and note that
\begin{equation}
  \langle I|\Phi_G\rangle
  =
  \begin{cases}
    1, & I=D,\\
    \alpha, & I\in\mathcal K(G),\\
    0, & \text{otherwise},
  \end{cases}
\end{equation}
which is efficient to compute. Indeed, checking $I=D$ is trivial while
checking $I\in\mathcal K(G)$ is an $O(n^2)$ operation. We note that the norm of the state $\|\Phi_G\|$ is by necessity unknown to the succinct description, as computing it would involve a sum over all exponentially many $I\subset [m]$.\par

Under our assumptions on the graph, we have two possibilities:
\begin{equation}
 |\Phi_G\rangle=
\begin{cases}
   |D\rangle & |\mathcal{K}(G)| = 0\\
    |D\rangle+\alpha|C\rangle & |\mathcal{K}(G)| = 1,
\end{cases}
\end{equation}
where $C$ is the unique $n$-multicolored clique distinct from $D$. Since $D$ and
$C$ have orthogonal one-particle subspaces, by
Lemma~\ref{lem:Superposition}, we have
\begin{equation}
  \OPT
  =
  \begin{cases}
    1, & |\mathcal{K}(G)| = 0,\\
    \alpha^2/(1+\alpha^2), & |\mathcal{K}(G)| = 1.
  \end{cases}
\end{equation}
Moreover, the optimum is attained uniquely, up to phase, by $\ket{D}$ in the
first case and by $\ket{C}$ in the second case. By our choice of $\alpha$, the promise
$\OPT\ge \lambda$ holds in both cases.\par 
At this point, we choose a constant $\varepsilon_\lambda$ as follows: 
\begin{equation}
\varepsilon_\lambda
  =
  \min\left\{
    \frac{7}{16},  \frac{\alpha^2-(\alpha-(\alpha-1)/4)^2}{1+\alpha^2}
  \right\}.
\end{equation}
This depends implicitly on $\lambda$ since the choice of $\alpha$ depended on $\lambda$. The reason for this particular choice will be evident in the next steps.

Now suppose, for contradiction, we have a randomized algorithm that outputs a Slater determinant $\ket{S}$ which satisfies Eq.~\eqref{eq:Scontradiction}, running in time $f(n)m^{o(n)}$ with success probability at least $0.99$. We will show that this algorithm can distinguish $|\mathcal{K}(G)| =0,1$, contradicting Assumption~\ref{ass:rETH}.\par

\begin{figure}
\centering
\begin{tikzpicture}[scale=0.8]
  \definecolor{pcA}{HTML}{F2A6A6}
  \definecolor{pcB}{HTML}{A9C8EE}
  \definecolor{pcC}{HTML}{A9DDBF}
  \definecolor{pcD}{HTML}{CBB4E8}
  \definecolor{pcE}{HTML}{F1D49B}
  \definecolor{cliquecol}{HTML}{23252F}
  \tikzset{
    delim/.style ={cliquecol, line width=1.4pt, line cap=round, line join=round},
    cle/.style   ={cliquecol, line width=0.8pt, line cap=round},
    ambe/.style  ={cliquecol!16, line width=0.6pt},
    cv/.style    ={circle, draw=cliquecol, line width=0.8pt, minimum size=7pt, inner sep=0pt},
    dv/.style    ={rectangle, draw=cliquecol, line width=0.8pt, minimum size=7pt, inner sep=0pt},
    fv/.style    ={circle, draw=cliquecol!35, line width=0.2pt, minimum size=6pt, inner sep=0pt, fill=cliquecol!12},
  }

  \node[anchor=east] at (2.40,0) {\Large $\ket{\Phi_G}\;=\;\alpha$};

  \draw[delim] (2.75,-0.92)--(2.75,0.92);
  \coordinate (C0) at (3.750,0.660); \coordinate (C1) at (3.122,0.204);
  \coordinate (C2) at (3.362,-0.534);\coordinate (C3) at (4.138,-0.534);
  \coordinate (C4) at (4.378,0.204);
  \coordinate (I0) at (3.750,0.10);  \coordinate (I1) at (3.970,-0.18);
  \coordinate (I2) at (3.510,-0.10); \coordinate (I3) at (3.800,-0.40);
  \coordinate (I4) at (3.700,0.40);
  \coordinate (O0) at (4.330,0.52);  \coordinate (O1) at (3.130,0.42);
  \coordinate (O2) at (4.470,-0.42); \coordinate (O3) at (3.200,-0.52);
  \foreach \a/\b in {I0/C0,I0/C2,I4/C0,I4/C1,I2/C1,I2/C2,I1/C4,I1/C3,I3/C2,I3/C3,
                     O0/C0,O0/C4,O1/C1,O2/C3,O2/C4,O3/C2,I0/I4,I2/I3,I1/O2,O1/I2}
    {\draw[ambe] (\a)--(\b);}
  \foreach \p in {I0,I1,I2,I3,I4,O0,O1,O2,O3}{\node[fv] at (\p) {};}
  \foreach \a/\b in {0/1,0/2,0/3,0/4,1/2,1/3,1/4,2/3,2/4,3/4}{\draw[cle] (C\a)--(C\b);}
  \node[cv, fill=pcA] at (C0) {}; \node[cv, fill=pcB] at (C1) {};
  \node[cv, fill=pcC] at (C2) {}; \node[cv, fill=pcD] at (C3) {};
  \node[cv, fill=pcE] at (C4) {};
  \draw[delim] (4.78,0.92)--(5.05,0)--(4.78,-0.92);

  \node at (5.78,0) {\Large $+$};

  \node[font=\small, cliquecol] at (3.75,-1.3) {$\ket{C}$};
  \node[font=\small, cliquecol] at (7.40,-1.3) {$\ket{D}$};
  \draw[delim] (6.50,-0.92)--(6.50,0.92);
  \coordinate (D0) at (7.400,0.600); \coordinate (D1) at (6.829,0.185);
  \coordinate (D2) at (7.047,-0.485);\coordinate (D3) at (7.753,-0.485);
  \coordinate (D4) at (7.971,0.185);
  \foreach \a/\b in {0/1,0/2,0/3,0/4,1/2,1/3,1/4,2/3,2/4,3/4}{\draw[cle] (D\a)--(D\b);}
  \node[dv, fill=pcA] at (D0) {}; \node[dv, fill=pcB] at (D1) {};
  \node[dv, fill=pcC] at (D2) {}; \node[dv, fill=pcD] at (D3) {};
  \node[dv, fill=pcE] at (D4) {};
  \draw[delim] (8.20,0.92)--(8.47,0)--(8.20,-0.92);
\end{tikzpicture}
\caption{\textbf{Encoding a clique into a superposition of Slaters.} The augmented graph $G$ is encoded as the unnormalized state $\ket{\Phi_G}=\alpha\sum_{C\in\mathcal{K}(G)}\ket{C} +\ket{D}$, where $\ket{D}$ is the fixed reference clique (right) and $\ket{C}$ is the unique hidden multicolored clique, if any (left). Locating the closest Slater to $\ket{\Phi_G}$ reveals whether $C$ exists.}
\label{fig:encoding}
\end{figure}

In the case $|\mathcal{K}(G)| =0$, we have $\ket{\Phi_G} = \ket{D}$ and $\OPT=1$, and hence
\begin{equation}
    |\langle S|D\rangle|^2 \geq 1-\varepsilon_\lambda.
\end{equation}
In particular, since $\varepsilon_\lambda\leq 7/16$, we have $|\langle S|D\rangle|>3/4.$ \par 

In the case $|\mathcal{K}(G)| =1$, by Lemma~\ref{lem:Superposition} we have
\begin{equation}
  |\langle S|D\rangle|+
  |\langle S|C\rangle|
  \le 1.
\end{equation}
Writing $x=|\langle S|D\rangle|$ for convenience, we have $|\langle S|C\rangle|\le 1-x$
and therefore
\begin{align}
  \frac{|\langle S|\Phi_G\rangle|}{\|\Phi_G\|}
  &\le
  \frac{x+\alpha(1-x)}{\sqrt{1+\alpha^2}} =\frac{\alpha-(\alpha-1)x}{\sqrt{1+\alpha^2}}
\end{align}
since $\ket{D}$ and $\ket{C}$ have disjoint occupied subspaces. If $x\ge 1/4$, then
\begin{equation}
  \frac{|\langle S|\Phi_G\rangle|^2}{\|\Phi_G\|^2}
  \le
  \frac{(\alpha-(\alpha-1)/4)^2}{1+\alpha^2}.
\end{equation}
Since $\alpha>1$, this is strictly smaller than the optimum value $\OPT=\frac{\alpha^2}{1+\alpha^2}.$ Define the quantity
\begin{equation}
  \Delta_\alpha
  \coloneqq
  \frac{\alpha^2-(\alpha-(\alpha-1)/4)^2}{1+\alpha^2}
  >
  0.
\end{equation}
Thus, in the case $|\mathcal{K}(G)|=1$, if
$|\langle S|D\rangle|\ge 1/4$ then $\ket{S}$ is at least $\Delta_\alpha$ below
optimal. Equivalently, if $\ket{S}$ is within $\Delta_\alpha$ of being optimal, then $|\langle S|D\rangle|<1/4$. Since $\varepsilon_\lambda\leq \Delta_\alpha$, if $\ket{S}$ is within $\varepsilon_\lambda$ of being optimal, then $|\langle S|D\rangle|<1/4$. 

\par 
Finally, we use this to distinguish the two cases as follows. Given the augmented graph $G$ and succinct classical description for $\ket{\Phi_G}$, we run the hypothetical closest-Slater algorithm and obtain $\ket{S}$. Since $D\subset G$ is explicitly known by construction, we can compute $|\langle S|D\rangle|$ efficiently to constant precision and declare ``no clique" if $|\langle S|D\rangle|>1/2$ and ``yes clique" otherwise. By the preceding arguments, this distinguishes the two possible cases. The total runtime is
\begin{equation}
  f(n)\bigl(n(p+1)\bigr)^{o(n)}
  =
  f'(n)p^{o(n)},
\end{equation}
contradicting Assumption~\ref{ass:rETH}. Hence no such algorithm exists, completing the proof.
\end{proof}

\subsection{Hardness in precision from the quantum separability problem}
\label{ref:hardnessepsilon}

The quantum learning algorithm described in Theorem~\ref{thm:quantum} brings the sample complexity down to $\poly(m, n, 1/\eps)$, but its runtime remains exponential.
In this section, we show that this inefficiency is inevitable, by proving that no quantum algorithm can solve the problem to $\eps=1/\poly(m)$ accuracy in $\poly(m)$ time even when there are only $n=2$ fermions, unless $\mathsf{NP}$ is in $\mathsf{BQP}$.
This holds true even if the algorithm is given the complete classical description of the input state $\rho$, and hence the hardness result also applies to the quantum copy access model since $\poly(m)$-dimensional quantum states can be prepared in $\poly(m)$ time from their classical descriptions.
\par 

\comphardthm*

We prove Theorem~\ref{thm:comp-hard} by reducing the task of finding the Slater determinant closest to a two-fermion state to the task of finding the product state closest to a bipartite state, also known as the \emph{quantum separability problem} \cite{gurvits2003classical,gharibian2010strong}.
Ref.~\cite{gharibian2010strong} proved that solving the quantum separability problem to inverse-polynomial accuracy is $\mathsf{NP}$-hard, and hence our hardness result follows.

The intuition behind the reduction can be understood as follows.
Consider two fermions in an even number of $m$ modes.
One can interpret these modes as $d=m/2$ spatial modes and a spin-$1/2$ degree of freedom.
For a Slater determinant in which the two fermions have opposite spins, the anti-symmetry requirement can always be satisfied by anti-symmetrizing the spin degree of freedom, and leaving the spatial wavefunction in an ordinary product state.
Hence, finding the closest Slater determinant is the same as finding the closest product state in the opposite-spin sector, where the spin label marks the bipartition, provided that Slater determinants built from wavefunctions that mix the two spin sectors cannot outperform opposite-spin ones.
We formalize this idea and prove Theorem~\ref{thm:comp-hard} in the following.

\begin{proof}[Proof of Theorem~\ref{thm:comp-hard}]
We begin by invoking the computational hardness of the quantum separability problem.
\begin{lemma}[Computational hardness of quantum separability, \cite{gharibian2010strong}]
\label{lem:quant-separable}
    There exists a constant $c_0>0$ such that the following promise problem is $\mathsf{NP}$-hard: given $M\in \mathbb{C}^{d^2\times d^2}, 0\preceq M\preceq I, M\neq 0, \gamma \in [0, 1]$, decide whether
    \begin{equation}
    	\max_{\substack{\ket{x}, \ket{y}\in \mathbb{C}^d\\\langle x|x\rangle=\langle y|y\rangle=1}} (\bra{x}\otimes\bra{y}) M(\ket{x}\otimes \ket{y}) \geq \gamma +1/d^{c_0}
    \end{equation}
    or 
    \begin{equation}
    	\max_{\substack{\ket{x}, \ket{y}\in \mathbb{C}^d\\\langle x|x\rangle=\langle y|y\rangle=1}} (\bra{x}\otimes\bra{y}) M(\ket{x}\otimes \ket{y}) \leq \gamma -1/d^{c_0}.
    \end{equation}
\end{lemma}

For simplicity, we define
\begin{equation}
	h_{\mathrm{Sep}}(M) = \max_{\substack{\ket{x}, \ket{y}\in \mathbb{C}^d\\\langle x|x\rangle=\langle y|y\rangle=1}} (\bra{x}\otimes\bra{y}) M(\ket{x}\otimes \ket{y}).
\end{equation}
Without loss of generality, we assume that $m$ is even and let $d=m/2$.
We label the single fermion Hilbert space as $\mathcal{H}_1 = \mathbb{C}^d \otimes \mathbb{C}^2$ and use $\ket{x^\sigma} = \ket{x}\otimes \ket{\sigma}, \ket{x}\in \mathbb{C}^d, \sigma\in \{\uparrow, \downarrow\}$ to denote the state vectors.
Let $\ket{e_i}, i=1, \ldots, d$ be the standard basis of $\mathbb{C}^d$.
Then the opposite-spin sector is the span of $\{\ket{e_i^{\uparrow}}\wedge \ket{e_j^\downarrow}\}_{i, j=1}^d$.
These $d^2$ vectors are orthonormal and hence the map
\begin{equation}
	J: \mathbb{C}^d \otimes \mathbb{C}^d \to \mathcal{H}_2 = \wedge^2(\mathbb{C}^{d}\otimes \mathbb{C}^2),
\end{equation}
satisfying
\begin{equation}
	J(\ket{e_i}\otimes \ket{e_j}) = \ket{e_i^\uparrow}\wedge \ket{e_j^\downarrow},
\end{equation}
is an isometry from the bipartite $d\times d$ dimension Hilbert space to the opposite-spin sector (Fig.~\ref{fig:isometry}).
Since $J$ is linear, we have $J(\ket{x}\otimes \ket{y}) = \ket{x^\uparrow}\wedge \ket{y^\downarrow}$.
Moreover, $\ket{x^\uparrow}$ and $\ket{y^\downarrow}$ are orthonormal whenever $\ket{x}$ and $\ket{y}$ are unit vectors, because the spin degree of freedom guarantees orthogonality.
Hence, $J$ maps product states to Slater determinants.
\par

\begin{figure}[t]
\centering
\resizebox{\columnwidth}{!}{%
\begin{tikzpicture}[
  font=\footnotesize,
  delim/.style ={cliquecol, line width=1.2pt, line cap=round, line join=round},
  spin/.style  ={-{Stealth[length=2.4mm,inset=0.5mm]}, graytone, line width=1.7pt, line cap=round},
  jarrow/.style={-{Stealth[length=3.4mm,inset=0.7mm]}, cliquecol, line width=1.9pt},
  lbl/.style   ={cliquecol},
]
\definecolor{greenacc}{HTML}{1F7A47}
\definecolor{blueacc}{HTML}{1F5C8F}
\definecolor{cliquecol}{HTML}{23252F}
\definecolor{graytone}{HTML}{6F6E68}
\newcommand{\orb}[4]{%
  \shade[top color=#3!50, bottom color=#3] (#1,#2) circle (0.36);
  \draw[cliquecol!35,line width=0.6pt] (#1,#2) circle (0.36);
  \node[text=white,font=\Large] at (#1,#2) {$#4$};%
}

\fill[cliquecol,opacity=0.04,rounded corners=8pt] (0.22,1.72) rectangle (3.38,3.48);
\fill[cliquecol,opacity=0.04,rounded corners=8pt] (4.60,1.72) rectangle (7.92,3.48);

\node[lbl,font=\small,text=cliquecol!80] at (1.80,3.76) {$\mathbb{C}^{d}\otimes\mathbb{C}^{d}$};
\node[lbl,font=\small,text=cliquecol!80] at (6.26,3.76) {$\wedge^{2}(\mathbb{C}^{d}\otimes\mathbb{C}^{2})$};

\draw[delim] (0.40,2.14)--(0.40,3.06);
\orb{0.86}{2.60}{greenacc}{x}
\draw[delim] (1.32,3.06)--(1.48,2.60)--(1.32,2.14);
\node[lbl,font=\Large] at (1.80,2.60) {$\otimes$};
\draw[delim] (2.12,2.14)--(2.12,3.06);
\orb{2.58}{2.60}{blueacc}{y}
\draw[delim] (3.04,3.06)--(3.20,2.60)--(3.04,2.14);

\draw[jarrow] (3.66,2.60)--(4.34,2.60);
\node[lbl,font=\LARGE] at (4.00,3.05) {$J$};

\draw[delim] (4.78,2.14)--(4.78,3.06);
\orb{5.24}{2.60}{greenacc}{x}
\draw[spin] (5.72,2.32)--(5.72,2.88);
\draw[delim] (5.80,3.06)--(5.96,2.60)--(5.80,2.14);
\node[lbl,font=\Large] at (6.26,2.60) {$\wedge$};
\draw[delim] (6.56,2.14)--(6.56,3.06);
\orb{7.02}{2.60}{blueacc}{y}
\draw[spin] (7.50,2.88)--(7.50,2.32);
\draw[delim] (7.58,3.06)--(7.74,2.60)--(7.58,2.14);
\end{tikzpicture}}
\caption{\textbf{Isometry to the opposite-spin sector. } The isometry $J$ sends a product state $\ket{x}\otimes\ket{y}$ in $\mathbb{C}^d\otimes\mathbb{C}^d$ to the Slater determinant $\ket{x^\uparrow}\wedge\ket{y^\downarrow}$ in the opposite-spin sector of $\wedge^2(\mathbb{C}^d\otimes\mathbb{C}^2)$.}
\label{fig:isometry}
\end{figure}
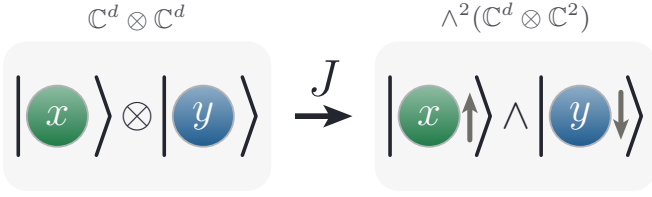

\par 

Now, we use the matrix $M$ from Lemma~\ref{lem:quant-separable} to construct a two-fermion state
\begin{equation}
	\rho_M = \frac{JMJ^\dagger}{\Tr(M)}.
\end{equation}
This is a valid $\binom{2d}{2}$-dimensional density matrix, because $0\preceq M\preceq I$ and $J$ is an isometry.
By construction, the Slater fidelity satisfies
\begin{equation}
	\mathrm{F}(\ket{x^\uparrow}\wedge \ket{y^\downarrow}) = \frac{1}{\Tr(M)} (\bra{x}\otimes \bra{y})M(\ket{x}\otimes \ket{y}).
\end{equation}
Therefore, we already have
\begin{equation}
    \mathsf{OPT} = \max_S \mathrm{F}(\ket{S})\geq \max_{x, y}\mathrm{F}(\ket{x^\uparrow}\wedge \ket{y^\downarrow}) = \frac{h_{\mathrm{Sep}}(M)}{\Tr(M)}.
\end{equation}

Next, we prove that $\mathsf{OPT}\leq \frac{h_{\mathrm{Sep}}(M)}{\Tr(M)}$ to establish their equality: $\mathsf{OPT}= \frac{h_{\mathrm{Sep}}(M)}{\Tr(M)}$.
To this end, note that any two-fermion Slater determinant can be written as $\ket{T}=\ket{w_1}\wedge\ket{w_2}$, where $\ket{w_1}, \ket{w_2}$ are orthonormal.
Furthermore, every single particle state vector $\ket{w_k}, k=1, 2$ can be decomposed into spin-up and spin-down parts
\begin{equation}
	\ket{w_k} = \ket{a_k^\uparrow} + \ket{b_k^\downarrow},
\end{equation}
where $\|\ket{a_k}\|^2 + \|\ket{b_k}\|^2=1, \ket{a_k}, \ket{b_k}\in\C^d$ are the unnormalized spatial wavefunctions.
From the linearity of the wedge product, we have
\begin{equation}
\begin{split}
    \ket{T} = &\ket{a_1^\uparrow}\wedge\ket{a_2^\uparrow} + \ket{b_1^\downarrow}\wedge\ket{b_2^\downarrow} \\
    &+\ket{a_1^\uparrow}\wedge\ket{b_2^\downarrow} + \ket{b_1^\downarrow}\wedge\ket{a_2^\uparrow}.
\end{split}
\end{equation}
Note that only the last two terms are in the opposite-spin sector, and hence the preimage of $\ket{T}$ under the map $J$ reads
\begin{equation}
\begin{split}
    J^\dagger \ket{T} &= J^\dagger (\ket{a_1^\uparrow}\wedge\ket{b_2^\downarrow} - \ket{a_2^\uparrow}\wedge\ket{b_1^\downarrow}) \\
    &= \ket{a_1}\otimes \ket{b_2}-\ket{a_2}\otimes \ket{b_1} \\
    &= \ket{\phi_T},
\end{split}
\end{equation}
where we have used the anti-symmetry of wedge product and defined $\ket{\phi_T} = \ket{a_1}\otimes \ket{b_2}-\ket{a_2}\otimes \ket{b_1}$.
Hence, we have
\begin{equation}
	\mathrm{F}(\ket{T}) = \frac{1}{\Tr(M)}\bra{\phi_T}M\ket{\phi_T}
\end{equation}
for any Slater determinant $\ket{T}$.
Note that 
\begin{equation}
\begin{split}
    &\bra{\phi_T}M\ket{\phi_T} = \|M^{1/2}\ket{\phi_T}\|^2\\
    &\leq (\|M^{1/2}\ket{a_1}\otimes \ket{b_2}\| + \|M^{1/2}\ket{a_2}\otimes \ket{b_1}\|)^2,
\end{split}
\end{equation}
via triangle inequality.
Also for any unnormalized vectors $\ket{a}, \ket{b} \in \C^d$, we have
\begin{equation}
	\|M^{1/2}\ket{a}\otimes \ket{b}\|^2\leq \|\ket{a}\|^2\|\ket{b}\|^2h_{\mathrm{Sep}}(M).
\end{equation}
Therefore, we have
\begin{equation}
	\bra{\phi_T}M\ket{\phi_T}\leq h_{\mathrm{Sep}}(M) (\|\ket{a_1}\|\|\ket{b_2}\|+\|\ket{a_2}\|\|\ket{b_1}\|)^2.
\end{equation}
Using the normalization condition $\|\ket{a_k}\|^2+\|\ket{b_k}\|^2=1, k=1, 2$, we may parameterize the norms as $\|\ket{a_k}\|=\cos(\theta_k), \|\ket{b_k}\|=\sin(\theta_k)$.
This means that
\begin{equation}
\begin{split}
    \|\ket{a_1}\|\|\ket{b_2}\|+\|\ket{a_2}\|\|\ket{b_1}\| &= \cos\theta_1\sin\theta_2+\cos\theta_2\sin\theta_1 \\
    &= \sin(\theta_1+\theta_2) \in [0, 1].
\end{split}
\end{equation}
Hence, we have 
\begin{equation}
	\mathrm{F}(\ket{T})\leq \frac{h_{\mathrm{Sep}}(M)}{\Tr(M)}
\end{equation}
for any Slater determinant $\ket{T}$ and therefore
\begin{equation}
	\mathsf{OPT}\leq \frac{h_{\mathrm{Sep}}(M)}{\Tr(M)}.
\end{equation}
Combined with the other direction, we have
\begin{equation}
	\mathsf{OPT}= \frac{h_{\mathrm{Sep}}(M)}{\Tr(M)}
\end{equation}

Now we are finally ready to prove the computational hardness.
For the sake of contradiction, assume that there is a polynomial-time quantum algorithm $\mathcal{A}$ that takes as input the classical description of any density matrix $\rho$ of $n=2$ fermions on $m$ modes and outputs a Slater determinant $\ket{S}$ satisfying
\begin{equation}
	\bra{S}\rho\ket{S}\geq \mathsf{OPT}-1/m^{c}
\end{equation}
with probability at least $2/3$.
In the following, we design a polynomial-time quantum algorithm that solves the quantum separability problem defined in Lemma~\ref{lem:quant-separable}.

Concretely, given a problem instance $(M\in \C^{d^2\times d^2}, \gamma)$ with problem size $d$ from the quantum separability problem, we set $m=2d$ and first compute the matrix
\begin{equation}
	\rho_M = \frac{JMJ^\dagger}{\Tr(M)}.
\end{equation}
This step is a $\poly(d)=\poly(m)$ dimensional matrix computation that takes $\poly(m)$ time.
Next, we run the algorithm $\mathcal{A}$ with  $\rho_M$ as the input, obtaining an output Slater determinant $\ket{T}$.
We orthogonalize and normalize it to obtain the spatial wavefunctions $\ket{a_1}, \ket{b_1}, \ket{a_2}, \ket{b_2}$ and calculate the preimage $\ket{\phi_T} = \ket{a_1}\otimes \ket{b_2}-\ket{a_2}\otimes \ket{b_1}$.
This step is standard matrix computation and only takes $\poly(d)=\poly(m)$ time.
Then we calculate the value
\begin{equation}
	v = \frac{\bra{\phi_T}M\ket{\phi_T}}{\Tr(M)} = \mathrm{F}(\ket{T}).
\end{equation}
We always have $v\leq \mathsf{OPT}$ and the guarantee of the learning algorithm $\mathcal{A}$ implies that
\begin{equation}
	\mathsf{OPT}-1/m^{c} \leq v\leq \mathsf{OPT}
\end{equation}
with probability at least $2/3$.
Since $\mathsf{OPT}=\frac{h_{\mathrm{Sep}}(M)}{\Tr(M)}$, we have that
\begin{equation}
	h_{\mathrm{Sep}}(M)-\Tr(M)/m^{c} \leq \Tr(M)v\leq h_{\mathrm{Sep}}(M).
\end{equation}
Moreover, $0\preceq M\preceq I$ implies that $\Tr(M) \in [0, d^2]$.
Therefore, we choose $c=c_0+2$ and obtain
\begin{equation}
	h_{\mathrm{Sep}}(M)-1/(4d^{c_0})\leq \Tr(M)v\leq h_{\mathrm{Sep}}(M).
\end{equation}
Finally, we calculate and compare the value of  $\Tr(M)v$ against $\gamma$, thus solving the quantum separability problem, which is $\mathsf{NP}$ hard in $d=m/2$ from Lemma~\ref{lem:quant-separable}.
The reduction is completely classical and runs in polynomial time.
Hence the existence of a quantum (or classical) learning algorithm that solves the Slater learning problem in $\poly(m)$ time implies that $\mathsf{NP}\subseteq \mathsf{BQP}$ (or $\mathsf{NP}\subseteq \mathsf{BPP}$), completing the proof of Theorem~\ref{thm:comp-hard}.
\end{proof}

\section{Extension to mixed states}\label{app:mixed}

In this appendix, we extend the results of the main text to mixed states. Throughout, $\rho$ denotes an arbitrary density matrix on the $n$-particle fermionic Hilbert space $\mathcal{H}_n$, $\mathrm{F}(S) = \langle S|\rho|S\rangle$, and $\OPT = \max_{\ket{S}\in \mathcal{S}} \langle S|\rho|S\rangle$. We first specify the access model we assume for the case of $\rho$ a mixed state. We then verify that the algorithmic upper and hardness lower bounds carry over. Finally, we prove the mixed-state analog of the $\mathrm{F}=2/3$ threshold. 

In the quantum access model, no modification is needed for an extension to mixed states, as we made no requirement that $\rho$ is pure. In the classical case, we assume the following.

\begin{definition}[Succinct classical access to a mixed state] A \textit{succinct classical description} of $\rho$ provides, in efficient time per sample, (i) for any few-body observable $A$, a sample of a random variable $X_A$ with $\mathbb{E}[X_A] = \Tr(A\rho)$ and $\mathrm{Var}[X_A] = O(1)$, and (ii) for any explicitly specified Slater $\ket{S}$, a sample of a random variable $Y_S \in[0,1]$ with $\mathbb{E}[Y_S]=\langle S|\rho|S\rangle$. 
\label{def:succ}
\end{definition}

This directly generalizes the pure-state access model of the main text. It is realized, for instance, by an ensemble representation $\rho = \sum_k p_k \ket{\psi_k}\!\bra{\psi_k}$ in which one can sample $k\sim p$ and each member $\ket{\psi_k}$ is provided with a pure-state succinct description. We can now state the following algorithms for the mixed state case. 

\begin{proposition}
    [Algorithmic results for mixed states] Let $\rho$ be an arbitrary mixed $n$-fermion state. Under Definition~\ref{def:succ}, the classical algorithm of Theorem~\ref{thm:upper} returns a Slater determinant with fidelity at least $\OPT -\varepsilon$, with probability at least $1-\delta$, in time $\exp\left[O\left(\frac{n^3}{\varepsilon}\log\frac{n}{\varepsilon}\right)\right]\cdot \poly(m,\log\frac{1}{\delta})$. Given quantum copies of $\rho$, the protocol of Theorem~\ref{thm:quantum} applies verbatim with the same sample complexity. 
\end{proposition}
\begin{proof}
Let us verify that each step of the proof of Theorem~\ref{thm:upper} applies for a general density matrix. First, we note that Lemmas~\ref{lem:active} and \ref{lem:noisy-active} are stated and proven for an arbitrary density matrix $\rho$, with $\gamma_{ij} = \Tr(c_j^\dagger c_i\rho)$ and $\tr \gamma = n$, so the construction of the active space $\mathcal{A}$ with $\dim \mathcal{A} = O(n^2/\varepsilon)$ and the bound $\OPT \leq \OPT_{\mathcal{A}} + \varepsilon/2$ require no change. Second, in Theorem~\ref{thm:upper}, $|\mathrm{F}(U)-\mathrm{F}(V)|\leq 2n\|U-V\|_F$ needs only that $\|\rho \| \leq 1$, which holds for a general $\rho$. Hence the covering net constructed with size $M = \exp\left[O\left(\frac{n^3}{\varepsilon}\log\frac{n}{\varepsilon}\right)\right]$ remains unchanged. Third, the exact evaluation of fidelities which was assumed in Theorem~\ref{thm:upper} is replaced by estimation under Definition~\ref{def:succ}. It suffices to straightforwardly estimate $\langle S_k|\rho|S_k\rangle$ for each net point to additive error $\varepsilon/8$ with failure probability $\delta/M$, which by Hoeffding's inequality costs $O(\varepsilon^{-2}\log(M/\delta))$ samples per point. A union bound gives a total failure probability $\delta$, and returning the empirical maximum over the net costs at most an additional $\varepsilon/4$ by our choice of covering net. By $\varepsilon/8 + \varepsilon/4 < \varepsilon/2$, we can then apply the conclusion of Theorem~\ref{thm:upper}. \par 

For the quantum protocol, we may note that the threshold search and Lemma~\ref{lem:agnosticQTS} are formulated for an \textit{arbitrary} state $\rho$ through the quantities $\Tr(\rho A_i)$, and the fermionic classical shadow estimation of $\gamma$ in Lemma~\ref{lem:noisy-active} makes no assumption on purity. Hence Theorem~\ref{thm:quantum} applies verbatim.
\end{proof}

Next, we remark that since pure states are a subset of mixed states, and the pure-state succinct classical description used by Theorem~\ref{thm:lower} fits under our mixed-state access model if we allow generalization to unnormalized $\rho$, and hence the computational hardness lower bounds transfer. In particular, the parameterized hardness lower bound of Theorem~\ref{thm:lower} applies to the mixed-state problem without modification and the proof of the sharpness of the 2/3 threshold in Theorem~\ref{thm:sharp} extends for the same reason. Moreover, we may note that the \textsf{NP}-hardness result of Theorem~\ref{thm:comp-hard} is already stated for mixed states. Indeed, the hard instances of this theorem are mixed states. \par

Next, we turn to the high-fidelity regime and the $\mathrm{F}=2/3$ threshold. Recall that Theorem~\ref{thm:certificate} showed the following: for a pure state $\ket{\psi}$ and Slater $\ket{G}$, if $|\langle G|\psi\rangle|^2 > 2/3$ and a local stationarity condition $B_G=0$ holds, then $\ket{G}$ is globally closest. Following a similar approach, let us fix a Slater $\ket{G}$ with occupied orbitals $g_{\mu}$ and unoccupied orbitals $f_a$, and let $\Pi_1^G$ denote the projector onto the span of the single particle-hole excitations $\ket{G_{a\mu}} = c_{f_a}^\dagger c_{g_\mu}\ket{G}$. We define the \textit{mixed residual}
\begin{equation}\label{eq:RG}
    R_G \coloneqq \|\Pi_1^G\rho\ket{G}\| = \sqrt{\sum_{a\mu} |\langle G_{a\mu}|\rho|G\rangle|^2}.
\end{equation}
The following Lemma generalizes Lemma~\ref{lem:brillouin} and identifies $R_G=0$ as the correct Brillouin condition for mixed states. 
\begin{lemma}
    [Mixed state Brillouin condition] The first variation of $\mathrm{F}(G) = \langle G|\rho|G\rangle$ over the manifold of Slater determinants vanishes at $\ket{G}$ if and only if $R_G=0$. 
\end{lemma}
\begin{proof}
    By Lemma~\ref{lem:brillouin}, the tangent space at $\ket{G}$ is spanned by the excitations $\ket{G_{a\mu}}$ (and trivially $i\ket{G}$) so that a path on the manifold passing through $\ket{G}$ at $\tau=0$ takes the form $\ket{G(\tau)} = \ket{G} + \tau\sum_{a\mu} T_{a\mu}\ket{G_{a\mu}} + O(\tau^2)$ up to phase. Substituting, we have
    \begin{equation}
        \langle G(\tau)|\rho|G(\tau)\rangle = \langle G|\rho | G\rangle + 2\tau \mathrm{Re}\sum_{a\mu} \overline{T_{a\mu}}\langle G_{a\mu}|\rho|G\rangle + O(\tau^2).
    \end{equation}
The linear term must vanish for arbitrary coefficients $T_{a\mu}$ in order for the first variation with respect to $\tau$ to vanish. This holds if and only if $\langle G_{a\mu}|\rho|G\rangle = 0$ for all $a,\mu$.
\end{proof}
We also point out that Lemma~\ref{lem:tail} (the technical result bounding higher-particle hole excitations using an inequality $q\leq \sqrt{2}(1-s)$) makes no reference to $\rho$ and hence applies here unchanged. An important next step will be a decomposition of $\rho$ into a pure state plus a remainder $\sigma$, formalized below. 

\begin{lemma}
    [Mixed state decomposition] Let $\rho\succeq 0$ and let $\ket{G}$ be a unit vector with $\alpha = \langle G|\rho|G\rangle >0$. Define $\ket{\phi} \coloneqq \rho\ket{G}/\sqrt{\alpha}$. Then 
    \begin{equation}
        \rho = \ket{\phi}\!\bra{\phi} + \sigma
    \end{equation}
    with $\sigma \succeq 0$ and $\sigma \ket{G} = 0$. 
    \label{lem:coherent}
\end{lemma}
\begin{proof}
    For any vector $\ket{x}$, Cauchy-Schwarz applied to the positive semidefinite bilinear form $\rho$ gives us $|\langle x|\rho |G\rangle|^2 \leq \langle x|\rho | x\rangle \langle G|\rho | G\rangle$, which is equivalent to $|\langle x|\phi\rangle|^2 \leq \langle x|\rho|x\rangle$ and proves that $\sigma = \rho - \ket{\phi}\!\bra{\phi}\succeq 0$. By definition of $\ket\phi$ we have $\langle \phi |G\rangle = \sqrt{\alpha}$, and hence $\sigma \ket{G} = \rho\ket{G}-\sqrt{\alpha}\ket{\phi}=0$.
\end{proof}

We are now in a position to establish the $\mathrm{F}=2/3$ threshold for mixed states. 
\begin{theorem}
    [$2/3$ certificate for mixed states] Let $\rho$ be an $n$-fermion density matrix on $\mathcal{H}_n$ and let $\ket{G}$ be a Slater determinant such that $\alpha = \langle G|\rho |G\rangle \geq 2/3$. Then every Slater determinant $\ket{H}$ satisfies    \begin{multline}\label{eq:hrhoh}
        \langle H |\rho|H\rangle \leq \alpha - \left(\frac{3\alpha}{2}-1\right)\left(1-|\langle G|H\rangle|^2\right) \\+ 2R_G\sqrt{1-|\langle G|H\rangle|^2}
\end{multline}
with $R_G$ defined in Eq.~\eqref{eq:RG}. In particular, if $\alpha>2/3$ and $R_G =0$, then $\ket{G}$ is the unique globally closest Slater to $\rho$, up to phase. Moreover, if $\alpha\geq 2/3+\Delta$ for some $\Delta >0$, then 
\begin{equation}\label{eq:opta23RG}
    \OPT \leq \alpha + \frac23 \frac{R_G^2}{\Delta}. 
\end{equation}
\end{theorem}

\begin{proof}
Let $\ket{H}$ be an arbitrary Slater and decompose it into particle-hole sectors relative to $\ket{G}$,
\begin{equation}
    \ket{H} = s\ket{G} + \ket{h_1} + \ket{h_{\geq 2}},
\end{equation}
with overall phase chosen so that $s = |\langle G|H\rangle|$. We set $p = \|\ket{h_1}\|^2$ and $q = \|\ket{h_{\geq 2}}\|$ so that $s^2+p+q^2=1$, and Lemma~\ref{lem:tail} then gives $q \leq \sqrt{2}(1-s)$. Using Lemma~\ref{lem:coherent}, we write 
\begin{equation}
\langle H|\rho|H\rangle = |\langle H|\phi\rangle|^2 + \langle H |\sigma |H\rangle,
\end{equation}
and we bound the two terms separately. Using the notation of Lemma~\ref{lem:coherent}, we decompose $\ket\phi$ as
\begin{equation}
    \ket\phi = \sqrt{\alpha}\ket G + \ket{\phi_1}+\ket{\phi_{\geq 2}}
\end{equation}
with $\|\ket{\phi_1}\|=R_G/\sqrt{\alpha}$. We also write $w = \|\ket{\phi_{\geq 2}}\|$ and $T = \Tr\sigma$. Due to the proliferation of variables, we relist some definitions: 
\begin{equation}
    \begin{aligned}
    \alpha &= \langle G|\rho|G\rangle,\,\,\, \ket\phi = \rho\ket{G}/\sqrt{\alpha}\quad \ket{G},\ket{H} \,\,\mathrm{Slaters},\\
    s &= |\langle G|H\rangle| ,\quad p =\|\ket{h_1}\|^2,\quad q = \|\ket{h_{\geq 2}}\|,\\
    w&= \|\ket{\phi_{\geq 2}}\|, \, \sigma = \rho - \ket{\phi}\!\bra{\phi},\,\,\,\quad T = \Tr\sigma.
    \end{aligned}
\end{equation}
The trace of $\rho$ can be decomposed as $\|\ket\phi\|^2 + T$ (note $\ket{\phi}$ need not be unit norm), or 
\begin{equation}\label{eq:arwt}
    \Tr \rho = \alpha + \frac{R_G^2}{\alpha} + w^2 +T = 1. 
\end{equation}
Since $R_G^2/\alpha \geq 0$ and $T\geq 0$, we have $w^2 \leq 1-\alpha$. In the regime $\alpha \geq 2/3$, we may write $1-\alpha \leq \alpha/2$ and hence $w^2\leq \alpha/2$. Hence the bound $q\leq \sqrt{2}(1-s)$ implies 
\begin{equation}
    wq \leq \sqrt{\alpha/2} \cdot \sqrt{2}(1-s) = \sqrt{\alpha}(1-s).
\end{equation}
It follows that $\sqrt{\alpha}s+wq\leq \sqrt{\alpha}$ in the regime $\alpha \geq 2/3$. Now, applying Cauchy-Schwarz to each particle-hole sector ($0,1,\geq 2$) gives 
\begin{equation}
    |\langle H|\phi\rangle| \leq \sqrt{\alpha}s + \frac{R_G}{\sqrt{\alpha}} \sqrt{p} + wq,
\end{equation}
while $\sigma \ket G = 0$ and $\|\sigma \| \leq \tr \sigma $ give $\langle H|\sigma |H\rangle = (\bra{h_1}+\bra{h_{\geq 2}}) \sigma (\ket{h_1} + \ket{h_{\geq 2}}) \leq T(p+q^2)$. Expanding the square, using $\sqrt{\alpha}s+wq \leq \sqrt{\alpha}$ on the cross term and $p\leq p+q^2$, and substituting using Eq.~\eqref{eq:arwt}, we obtain
\begin{equation}
    \begin{aligned}
        &\langle H|\rho|H\rangle \\
        &\leq \left(\sqrt{\alpha}s+wq\right)^2 + 2R_G \sqrt{p} + \frac{R_G^2}{\alpha} p + T(p+q^2)\\
        &\leq \left(\sqrt{\alpha}s+wq\right)^2 +(1-\alpha-w^2)(p+q^2)+2R_G\sqrt{p}.
    \end{aligned}
\end{equation}
We claim that the first two terms of the right-hand side satisfy
\begin{multline}\label{eq:biginequality}
    \left(\sqrt{\alpha}s+wq\right)^2 +(1-\alpha-w^2)(p+q^2) \\
    \leq \alpha - \left(\frac{3\alpha}{2}-1\right)(p+q^2).
\end{multline}
Together with the observation that since $p+q^2=1-s^2$ we have $\sqrt{p}\leq \sqrt{1-s^2}$, this would establish Eq.~\eqref{eq:hrhoh}. 
\par 
Let us now prove Eq.~\eqref{eq:biginequality}. Expanding the left-hand side and substituting $s^2=1-p-q^2$ reduces the above to $2\lambda sq-\lambda^2p\leq \frac12(p+q^2)$, where $\lambda \coloneqq w/\sqrt{\alpha} \in [0,1/\sqrt{2}]$. The left-hand side is concave in $\lambda$. If its maximum over $[0,1/\sqrt{2}]$ is attained at an interior stationary point, then $p>0$ and $\lambda = sq/p\leq 1/\sqrt{2}$, so the maximum value is $s^2q^2/p \leq p/2\leq \frac12(p+q^2)$. Otherwise, the maximum is attained at $\lambda = 1/\sqrt{2}$ with value $\sqrt{2}sq-p/2$. Squaring the inequality $s + q/\sqrt{2}\leq 1$ and substituting $p=1-s^2-q^2$ gives $\sqrt{2}sq\leq p+q^2/2$, so the value is again at most $\frac12 (p+q^2)$. This proves the inequality $2\lambda sq-\lambda^2p\leq \frac12(p+q^2)$ and hence Eq.~\eqref{eq:biginequality}, which is equivalent, and establishes Eq.~\eqref{eq:hrhoh}. \par

Moreover, if $\alpha>2/3$ and $R_G = 0$, then Eq.~\eqref{eq:hrhoh} gives $\langle H|\rho|H\rangle < \alpha$ whenever $|\langle G|H\rangle|<1$. This proves uniqueness of the global optimum. \par 

Finally, setting $X = \sqrt{1-s^2}$ and $\mu = \frac{3\alpha}{2}-1$, the right-hand side of Eq.~\eqref{eq:hrhoh} is at most $\alpha + \sup_{X\geq 0} (-\mu X^2 + 2R_G X) = \alpha + R_G^2/\mu$, and $\alpha\geq 2/3+\Delta$ gives $\mu \geq 3\Delta/2$, which establishes Eq.~\eqref{eq:opta23RG}, concluding the proof. 
\end{proof}
Since pure states are a subset of mixed states, the construction of Theorem~\ref{thm:sharp} shows that the threshold $2/3$ cannot be lowered. In conclusion, each result of the main text holds for mixed states upon replacing $|\langle S|\psi\rangle|^2$ with $\langle S|\rho|S\rangle$ wherever needed, replacing $B_G$ with $R_G$, and using the classical access model in Definition~\ref{def:succ}.

\section{Technical lemmas}
\label{app:lemmas}

We collect various lemmas used in the proofs in this work.

\begin{lemma}[Truncating from the 1-RDM]
\label{lem:active}
Let $\rho$ be a density matrix in $\mathcal{H}_n$ and let
$\gamma$ be its one-body density matrix $\gamma_{ij}=\operatorname{tr}(\rho\,c_j^\dagger c_i)$. Let $\mathcal A\subseteq \mathcal H_1$ be a subspace with projector
$P_{\mathcal A}$ and $\dim\mathcal A\ge n$. Fix $\varepsilon>0$ and suppose that every orbital
outside $\mathcal A$ has occupation at most $\varepsilon/n$,
\begin{equation}
  \|(I-P_{\mathcal A})\gamma(I-P_{\mathcal A})\| \le \varepsilon/n.
\end{equation}
Let $\mathcal S_{\mathcal A}$ denote the set of Slater determinants whose
occupied orbitals all lie in $\mathcal A$ and define
\begin{equation}
  \OPT_{\mathcal A}
  :=
  \max_{|S\rangle\in\mathcal S_{\mathcal A}}
  \langle S|\rho|S\rangle .
\end{equation}
Then
\begin{equation}
  \OPT
  \le
  \OPT_{\mathcal A}+\varepsilon.
\end{equation}
\end{lemma}

\begin{proof}
Let $|G\rangle$ be an arbitrary $n$-particle Slater determinant. Let $\mathcal L\subseteq\mathcal H_1$ be the occupied one-particle
subspace of $|G\rangle$.

After a change of basis in $\mathcal{L}$, we may write
\begin{equation}
  |G\rangle
  =
  \prod_{\alpha=1}^n(
\cos\theta_\alpha\,a_\alpha^\dagger
  +
\sin\theta_\alpha\,b_\alpha^\dagger )|0\rangle .
\end{equation}
where the $a_\alpha^\dagger$ and $b_\alpha^\dagger$ create orthonormal orbitals in $\mathcal{A}$ and $\mathcal{A}^\perp$, respectively. In words, the procedure starts with a unit vector $\ell_1 \in \mathcal{L}$ and finds a unit vector $a_1 \in \mathcal{A}$ which maximizes $|\langle \ell_1|a_1\rangle|$, calling this angle $\theta_1$. One then considers the complement of $\ell_1$ in $\mathcal{L}$ and the complement of $a_1$ in $\mathcal{A}$ and repeats the procedure. The $b_\alpha^\dagger$ create orbitals corresponding to $\ell_\alpha - \langle \ell_\alpha|a_\alpha\rangle a_\alpha$, if this is nonzero, and $b_\alpha = 0$ otherwise. 

\par 

Expanding the product in $\ket{G}$ produces 
\begin{equation}
  |G\rangle
  =
  c\,|G_{\mathcal A}\rangle
  +
  \sqrt{1-c^2}\,|K\rangle,
\end{equation}
where $c=\prod_{\alpha=1}^n\cos\theta_\alpha$,  $|G_{\mathcal A}\rangle
=
a_1^\dagger\cdots a_n^\dagger|0\rangle$ is the Slater determinant whose occupied orbitals lie entirely in $\mathcal{A}$, and $|K\rangle$ is in the span of Slater determinants with at
least one particle in
\begin{equation}
  \mathcal W
  :=
  \Span\{b_\alpha:\sin\theta_\alpha\ne0, \alpha =1,\ldots, n\}.
\end{equation}
Let $P_{\mathcal W}$ be the one-particle projector onto $\mathcal W$. Since $\mathcal{W}\subset\mathcal{A}^\perp$, and every orbital in $\mathcal A^\perp$ has occupation at most
$\varepsilon/n$, we have
\begin{equation}
  \operatorname{tr}(P_{\mathcal W}\gamma)
  \le
  \dim(\mathcal W)\varepsilon/n
  \le
\varepsilon.
\end{equation}
Next, choose a basis $w_1,\ldots, w_r$ of $\mathcal{W}$ and a basis $u_1,\ldots, u_{m-r}$ of $\mathcal{W}^\perp$. Let 
\begin{equation}
    n_{\mu} = c_{w_\mu}^\dagger c_{w_\mu}
\end{equation}
be the occupation-number operator for $w_\mu$. Let $\widehat Q_{\mathcal W} = 1-\prod_{\mu=1}^r(1-n_\mu)$ be the many-body projector onto the subspace with at least one particle in $\mathcal W$, and let $\widehat N_{\mathcal W} =\sum_{\mu=1}^r n_\mu$ be
the number operator in $\mathcal W$. In a state with $\widehat{N}_{\mathcal{W}} = k$, we have $\widehat Q_{\mathcal{W}} = 1$ if $k\geq 1$ and $\widehat Q_{\mathcal{W}}=0$ if $k=0$. Hence
\begin{equation}
  \widehat Q_{\mathcal W}\preceq \widehat N_{\mathcal W},
\end{equation}
where $A\preceq B$ means that $B-A$ is positive semidefinite. Thus we have
\begin{equation}
  \operatorname{tr}(\rho \widehat Q_{\mathcal W})
  \leq
  \operatorname{tr}(\rho \widehat N_{\mathcal W})
  =
  \operatorname{tr}(P_{\mathcal W}\gamma).
\end{equation}
In the last step, we have used
\begin{equation}
    \sum_{\mu=1}^r \tr(\rho c_{w_\mu}^\dagger c_{w_\mu}) = \sum_{\mu=1}^r\sum_{ij} (w_\mu)_i^*(w_\mu)_j \tr(\rho c_j^\dagger c_i).
\end{equation}
Thus
\begin{equation}
  \operatorname{tr}(\rho \widehat Q_{\mathcal W})\le \varepsilon,
\end{equation}
i.e. the overlap of $\rho$ with the subspace with at least one particle in $\mathcal{W}$ is upper bounded by $\varepsilon$. 
Since $|K\rangle$ is supported entirely on this subspace ($Q_{\mathcal W}|K\rangle=|K\rangle$) we have
\begin{equation}\label{eq:KrhoKeps}
  \langle K|\rho|K\rangle
  =
  \langle K|Q_{\mathcal W}\rho Q_{\mathcal W}|K\rangle
  \le
  \operatorname{tr}(\rho Q_{\mathcal W})
  \le
  \varepsilon.
\end{equation}
Next, we expand
\begin{align}
  \langle G|\rho|G\rangle
  &=
  c^2\langle G_{\mathcal A}|\rho|G_{\mathcal A}\rangle
  +(1-c^2)\langle K|\rho|K\rangle \nonumber\\
  &\quad
  +2c\sqrt{1-c^2}\,
  \operatorname{Re}\langle G_{\mathcal A}|\rho|K\rangle.
\end{align}
Since $\rho\succeq0$ defines a Hermitian form, Cauchy--Schwarz gives
\begin{equation}
  |\langle G_{\mathcal A}|\rho|K\rangle|
  \le
  \sqrt{
  \langle G_{\mathcal A}|\rho|G_{\mathcal A}\rangle
  \langle K|\rho|K\rangle}
\end{equation}
and hence 
\begin{equation}
    \langle G|\rho|G\rangle
  \le
  \left(
    c\sqrt{\langle G_{\mathcal A}|\rho|G_{\mathcal A}\rangle}
    +
    \sqrt{1-c^2}\sqrt{\langle K|\rho|K\rangle}
  \right)^2.
\end{equation}
Next, we use $\langle G_{\mathcal A}|\rho|G_{\mathcal A}\rangle\leq \OPT_{\mathcal{A}}$ and Eq.~\eqref{eq:KrhoKeps} to write 
\begin{equation}
    \langle G|\rho|G\rangle
  \le\left(
    c\sqrt{\OPT_{\mathcal A}}
    +
    \sqrt{1-c^2}\sqrt{\varepsilon}
  \right)^2.
\end{equation}
Next, applying Cauchy-Schwarz to the real vectors $(c,\sqrt{1-c^2})$ and $(\sqrt{\OPT_{\mathcal A}}, \sqrt{\varepsilon})$ gives us
\begin{align}
  \langle G|\rho|G\rangle
  &\le
  \OPT_{\mathcal A}+\varepsilon.
\end{align}
Since $\ket{G}$ was an arbitrary Slater, 
\begin{equation}
  \OPT
  \le
  \OPT_{\mathcal A}+\varepsilon.
\end{equation}
\end{proof}

The estimator for the 1-RDM $\gamma$ may be noisy. The next lemma shows that the active-space truncation of Lemma \ref{lem:active} is robust to noise. The key idea is that only the subspace $\mathcal{A}$, not a particular choice of basis of $\mathcal{A}$, needs to be moderately robust to noise in order for the truncation of the 1-RDM to be useful.

\begin{lemma}[Truncating from a noisy 1-RDM]
\label{lem:noisy-active}
In the setting of Lemma~\ref{lem:active}, let $\hat\gamma$ be a Hermitian matrix
estimating the 1-RDM $\gamma$ with error $\|\hat\gamma-\gamma\|\le\eta$. Fix a threshold $\hat\tau>\eta$ and let $\mathcal A$ be the span of the eigenvectors $\hat\psi^a$ of $\hat \gamma$ whose eigenvalues exceed
$\hat\tau$, i.e. $ \mathcal A = \operatorname{span}\{\hat\psi^a \mid \hat\nu_a>\hat\tau\}$.
Let $P_{\mathcal A}$ be a projector onto $\mathcal{A}$. Then $\mathcal A$ satisfies the hypothesis of
Lemma~\ref{lem:active} with
\begin{equation}
  \|(I-P_{\mathcal A})\gamma(I-P_{\mathcal A})\|\le\hat\tau+\eta
\end{equation}
and $\dim\mathcal A\le\frac{n}{\hat\tau-\eta}$. In particular, taking $\eta=\varepsilon/4n$ and $\hat\tau=3\varepsilon/4n$ gives
$\|(I-P_{\mathcal A})\gamma(I-P_{\mathcal A})\|\le\varepsilon/n$ and
$\dim\mathcal A\le 2n^2/\varepsilon$, so that
$\OPT\le\OPT_{\mathcal A}+\varepsilon$ with
$\dim\mathcal A=O(n^2/\varepsilon)$.\par 

Furthermore, an estimator $\hat\gamma$ satisfying $\|\hat\gamma-\gamma\|\leq \eta$ with probability at least $1-\delta$ can be obtained from 
\begin{equation}\label{eq:nfcs}
    n_{\mathrm{fcs}} = O\left(\frac{m^2\log(m/\delta)}{\eta^2}\right)
\end{equation}
quantum copies of $\rho$ via fermionic classical shadows~\cite{zhao2021fermionic}.
\end{lemma}
\begin{proof}
Let $v$ be a unit vector outside $\mathcal A$. Then $v$ lies in the span of
the eigenvectors of $\hat\gamma$ with eigenvalue at most $\hat\tau$, so
$\langle v|\hat\gamma|v\rangle\le\hat\tau$. Hence
\begin{equation}
  \langle v|\gamma|v\rangle
  = \langle v|\hat\gamma|v\rangle - \langle v|(\hat\gamma-\gamma)|v\rangle
  \le \hat\tau+\eta.
\end{equation}
Since $\gamma\succeq 0$, the operator norm of
$(I-P_{\mathcal A})\gamma(I-P_{\mathcal A})$ is the maximum of
$\langle v|\gamma|v\rangle$ over unit vectors $v$ outside $\mathcal A$, and so
$\|(I-P_{\mathcal A})\gamma(I-P_{\mathcal A})\|\le\hat\tau+\eta$.

Next we bound the dimension of $\mathcal{A}$. The sorted eigenvalues of
$\hat\gamma$ and $\gamma$ satisfy $\hat\nu_a\le\nu_a+\eta$, so every
$\hat\nu_a>\hat\tau$ has $\nu_a>\hat\tau-\eta$. The number of such orbitals is
at most $\operatorname{tr}\gamma/(\hat\tau-\eta)=n/(\hat\tau-\eta)$, since
$\gamma\succeq 0$ and $\operatorname{tr}\gamma=n$. Hence
$\dim\mathcal A\le n/(\hat\tau-\eta)$. Both bounds hold for the stated $\eta$ and $\hat\tau$, and
$\OPT\le\OPT_{\mathcal A}+\varepsilon$ then follows from
Lemma~\ref{lem:active}.
\par 

We next justify Eq.~\eqref{eq:nfcs}. The fermionic classical shadow formalism provides single-shot estimators $\hat\gamma^{(s)}$ which are Hermitian $m\times m$ matrices satisfying $\mathbb{E}[\hat\gamma^{(s)}] = \gamma$. We consider averaging $N$ of them to obtain an estimator $\hat\gamma \coloneqq \frac1N\sum_{s=1}^N\hat\gamma^{(s)}$. By Ref.~\cite{zhao2021fermionic} (in particular Eq. (13) and ensuing discussion), we have the ``shadow norm" 
\begin{equation}
    \|\hat\gamma^{(s)}\| \leq 1/\lambda_{m,1} \equiv \frac{\binom{2m}{2}}{\binom{m}{1}} = 2m-1.
\end{equation}
We next apply a matrix concentration inequality~\cite{tropp2012user,tropp2015introduction}. Defining centered snapshots $X_s \coloneqq \hat\gamma^{(s)}-\gamma$, which satisfy $\mathbb{E}[X_s] = 0$, $\|X_s\| \leq 2m$, and $\|\mathbb{E}[X_s^2]\|\leq \mathbb{E}\|X_s\|^2 \leq (2m)^2$. Then the matrix Bernstein inequality yields 
\begin{equation}
    \mathrm{Pr}\left[\|\hat\gamma-\gamma\| \geq \eta \right]\leq 2m \exp\left(-\frac{cN\eta^2}{m^2}\right) 
\end{equation}
for a universal $O(1)$ constant $c>0$. In words, the probability that $\|\hat\gamma-\gamma\|\geq \eta$ decays exponentially with $N\eta^2/m^2$. Hence $\|\hat\gamma-\gamma\|\leq \eta$ holds with probability at least $1-\delta$ as long as $N = O(m^2\log(m/\delta)/\eta^2)$, as was to be shown. 
\end{proof}

The following lemma is a useful technical result on the closest Slater to a superposition of two Slaters. 
\begin{lemma}[Superpositions of Slaters]
\label{lem:Superposition}
Let $|E\rangle$ and $|F\rangle$ be $n$-particle Slater determinants
whose occupied one-particle subspaces $\mathcal E,\mathcal F\subseteq\mathcal H_1$
are orthogonal. Assume $n\geq 2$. Then for every $n$-particle Slater
determinant $|G\rangle$, we have
\begin{equation}
  |\langle E|G\rangle|+|\langle F|G\rangle|\le 1.
\end{equation}
Moreover, consider a superposition of the form
\begin{equation}\label{eq:psiEF}
|\psi\rangle=\frac{|E\rangle+\alpha|F\rangle}{\sqrt{1+|\alpha|^2}}
\end{equation}
for $|\alpha|>1$. The unique Slater determinant with maximal fidelity with $\ket{\psi}$ is $|F\rangle$ (up to phase), with fidelity $|\alpha|^2/(1+|\alpha|^2)$.
\end{lemma}

\begin{proof}
Choose orthonormal bases $e_1,\ldots,e_n$ of $\mathcal E$ and
$f_1,\ldots,f_n$ of $\mathcal F$, so that
\begin{equation}
  |E\rangle=c_{e_1}^\dagger\cdots c_{e_n}^\dagger|0\rangle,\qquad
  |F\rangle=c_{f_1}^\dagger\cdots c_{f_n}^\dagger|0\rangle.
\end{equation}
Let $|G\rangle$ be an arbitrary Slater and choose an orthonormal basis $g_1, \ldots, g_n$ such that $\ket{G} = c_{g_1}^\dagger\cdots c_{g_n}^\dagger|0\rangle$. Define the overlap matrices
$M^{(E|G)}_{\mu\alpha}=\langle e_\mu|g_\alpha\rangle$,
$M^{(F|G)}_{\mu\alpha}=\langle f_\mu|g_\alpha\rangle$, so
\begin{equation}
  \langle E|G\rangle=\det M^{(E|G)},\quad
  \langle F|G\rangle=\det M^{(F|G)}.
\end{equation}
Since $\mathcal E\perp\mathcal F$,
each unit vector $g_\alpha$ satisfies
\begin{equation}
\sum_{\mu=1}^n \left( |\langle e_\mu|g_\alpha\rangle|^2 + |\langle f_\mu|g_\alpha\rangle|^2\right) \leq 1. 
\end{equation}
Summing over $\alpha$ gives 
\begin{equation}
  \sum_{\mu,\alpha=1}^n \left( |M^{(E|G)}_{\mu\alpha}|^2 +  |M^{(F|G)}_{\mu\alpha}|^2 \right) \leq n. 
\end{equation} 

We define the average weights 
\begin{equation}
u_E = \frac1n \sum_{\mu,\alpha} |M_{\mu\alpha}^{(E|G)}|^2,\quad u_F= \frac1n \sum_{\mu,\alpha} |M_{\mu\alpha}^{(F|G)}|^2.
\end{equation}
which satisfy $u_E + u_F \leq 1$ and $u_E,u_F \geq 0$. Next, we define singular values $\{s_\alpha^A\}_{\alpha=1}^n$ for $M^{(A|G)}$ for $A=E,F$, such that
\begin{equation}
|\det M^{(A|G)}| = \prod_{\alpha=1}^n |s_\alpha^A|.
\end{equation}
In terms of these singular values, we have
\begin{equation}
u_A = \frac1n \sum_{\alpha = 1}^n |s_\alpha^A|^2.
\end{equation}
Next, we apply the AM--GM inequality to the set of numbers $\{|s_\alpha^A|^2\}$, yielding
\begin{equation}
  \left(\prod_{\alpha=1}^n |s_\alpha^A|^2\right)^{1/n}
  \le
  \frac1n\sum_{\alpha=1}^n |s_\alpha^A|^2
  =
  u_A. 
\end{equation}
Raising both sides to the power of $n/2$ yields $|\det M^{(A|G)}| \leq u_A^{n/2}$. Hence $|\langle A|G\rangle| \leq u_A^{n/2} \leq u_A$ for $A=E,F$, where in the last step we have used that $n\geq 2$ and $0\leq u_A \leq 1$. Finally, we conclude that
\begin{equation}
|\langle E|G\rangle| + |\langle F|G\rangle| \leq u_E+u_F \leq 1,
\end{equation}
as was to be shown. This bounds the overlap of an arbitrary Slater $\ket{G}$ with an orthogonal pair of Slaters $\ket{E}$ and $\ket{F}$. 
\par 

Next, we consider the state $\ket{\psi}$ defined in Eq.~\eqref{eq:psiEF}. Let $\ket{G}$ again be an arbitrary Slater and write $x_E=|\langle E|G\rangle|$ and
$x_F=|\langle F|G\rangle|$. Without loss of generality, we assume $\alpha >0$, so that
\begin{equation}
|\langle G |\psi\rangle | \leq \frac{x_E + \alpha x_F}{\sqrt{1+\alpha^2}}.
\end{equation}
Because $\alpha>1$ by assumption, the numerator is maximized over the domain 
$\{x_E,x_F\ge0:\ x_E+x_F\le1\}$ at precisely $(x_E,x_F)=(0,1)$, yielding
\begin{equation}
|\langle G|\psi\rangle| \leq \frac{\alpha}{\sqrt{1+\alpha^2}}.
\end{equation}
This bound is attained when
$x_F=|\langle F|G\rangle|=1$, which means $\ket{G} = \ket{F}$ as $n$-particle
Slater determinants, up to an overall phase. Thus $\ket{F}$ is the unique
closest Slater, up to phase.
\end{proof}

The following Lemma identifies single particle-hole excitations as the tangent directions of the manifold of Slater determinants, following Refs.~\cite{Thouless1960,SzaboOstlund1996}. The stationarity condition for the fidelity mirrors the classic Brillouin condition for the stationarity of the variational Hartree-Fock energy~\cite{Brillouin1933}, and hence we refer to the following as the \textit{Brillouin condition for fidelity}. 

\begin{lemma}[Brillouin condition for fidelity]
\label{lem:brillouin}
Let
\begin{equation}
  |G\rangle=c_{g_1}^\dagger\cdots c_{g_n}^\dagger|0\rangle
\end{equation}
be a Slater determinant with occupied orbitals $g_1,\ldots, g_n$ and
unoccupied orbitals $\{f_a\}$ with corresponding creation operators $c_{f_a}^\dagger$, for $a=1,\ldots, m-n$. Let
\begin{equation}
  |G_{a\alpha}\rangle=c_{f_a}^\dagger c_{g_\alpha} |G\rangle
\end{equation}
denote the single particle-hole excitations on top of $\ket{G}$.  For a normalized $n$-fermion
state $|\psi\rangle$, define
\begin{equation}
  \mathrm{F}(G)=|\langle G|\psi\rangle|^2,
  \qquad
  B_G^2=\sum_{a,\alpha}|\langle G_{a\alpha}|\psi\rangle|^2 \geq 0.
\end{equation}
Then the first variation of $\mathrm{F}$ at $\ket{G}$ vanishes if and only if
\begin{equation}
  \langle G|\psi\rangle\, \langle G_{a\alpha}|\psi\rangle=0
 \end{equation}
 for all $a,\alpha$. 
In particular, if $\mathrm{F}(G) >0$, then $G$ is a
stationary point if and only if 
\begin{equation}
B_G =0.
\end{equation}

\end{lemma}

\begin{proof}
We first establish some basic facts about the tangent space of the
manifold of Slater determinants near $\ket{G}$ (referring to e.g. Ref.~\cite{edelman1998geometry} for further details). The tangent space is the
space of derivatives, at $\ket{G}$, of smooth paths in the manifold that pass through $\ket{G}$. Thus we consider a path
\begin{equation}
  \ket{G(\tau)}
  =
  c_{g_1}^\dagger(\tau)\cdots c_{g_n}^\dagger(\tau)\ket{0},
\end{equation}
where $c_{g_\alpha}^\dagger(0)=c_{g_\alpha}^\dagger$, and where the
orbitals $g_1(\tau),\ldots,g_n(\tau)$ remain orthonormal for all
$\tau$.  In a fixed one-particle basis, each creation operator has the
form
\begin{equation}
  c_{g_\alpha}^\dagger(\tau)
  =
  \sum_i \mathcal G_{\alpha i}(\tau)c_i^\dagger .
\end{equation}
The derivative of $c_{g_\alpha}^\dagger(\tau)$ is therefore:
\begin{equation}
  \dot c_{g_\alpha}^\dagger
  :=
  \left.\frac{\dd}{\dd\tau}c_{g_\alpha}^\dagger(\tau)\right|_{\tau=0}
  =
  \sum_i \mathcal G'_{\alpha i}(0)c_i^\dagger .
\end{equation}
Since the occupied
orbitals $g_\alpha$ together with the unoccupied orbitals $f_a$ form an
orthonormal basis of the one-particle Hilbert space, we may write
\begin{equation}\label{eq:cAT}
  \dot c_{g_\alpha}^\dagger
  =
  \sum_{\beta=1}^n A_{\beta\alpha}c_{g_\beta}^\dagger
  +
  \sum_a T_{a\alpha}c_{f_a}^\dagger .
\end{equation}
Orthonormality for all $\tau$ imposes the following constraint. Since $ \langle g_\beta(\tau)|g_\alpha(\tau)\rangle
  =
  \delta_{\beta\alpha}$, we have $\langle \dot g_\beta|g_\alpha\rangle
  +
  \langle g_\beta|\dot g_\alpha\rangle
  =
  0$, and hence taking $A_{\beta\alpha}=\langle g_\beta|\dot g_\alpha\rangle$, this becomes
\begin{equation}
  A^\dagger=-A,
\end{equation}
i.e. $A$ is anti-Hermitian.
Next, differentiating the many-body Slater determinant gives us
\begin{equation}
  \left.\frac{\dd}{\dd\tau}\ket{G(\tau)}\right|_{\tau=0}
  =
  \sum_{\alpha=1}^n
  c_{g_1}^\dagger\cdots
  \dot c_{g_\alpha}^\dagger
  \cdots
  c_{g_n}^\dagger\ket{0}.
\end{equation}
Substituting $\dot c_{g_\alpha}^\dagger$ gives one term involving $A$ and one term involving $T$. The term involving $A$ has the form
\begin{equation}
  \sum_{\alpha,\beta}
  A_{\beta\alpha}
  c_{g_1}^\dagger\cdots
  c_{g_\beta}^\dagger
  \cdots
  c_{g_n}^\dagger\ket{0},
\end{equation}
where $c_{g_\beta}^\dagger$ replaces $c_{g_\alpha}^\dagger$.
If $\beta\ne\alpha$, this contains two identical creation
operators and hence vanishes by fermionic antisymmetry.  Thus only the
terms with $\beta=\alpha$ survive, giving
\begin{equation}
  \sum_{\alpha=1}^n A_{\alpha\alpha}\ket{G}
  =
  \tr(A)\ket{G}.
\end{equation}
Since $A$ is anti-Hermitian, $\tr(A)$ is purely imaginary. This term corresponds to infinitesimal phase rotation of $\ket{G}$.
The term involving $T$ gives
\begin{equation}
  \sum_{a,\alpha}
  T_{a\alpha}\,
  c_{g_1}^\dagger\cdots
  c_{f_a}^\dagger
  \cdots
  c_{g_n}^\dagger\ket{0}
\end{equation}
where $c_{f_a}^\dagger$ replaces $c_{g_\alpha}^\dagger$. This is a single particle-hole excitation
\begin{equation}
  c_{g_1}^\dagger\cdots
  c_{f_a}^\dagger
  \cdots
  c_{g_n}^\dagger\ket{0}
  =
  c_{f_a}^\dagger c_{g_\alpha}\ket{G}
  =
  \ket{G_{a\alpha}},
\end{equation}
up to the sign convention fixed by the ordering of the creation
operators. Therefore every tangent vector has the form
\begin{equation}
  \left.\frac{\dd}{\dd\tau}\ket{G(\tau)}\right|_{\tau=0}
  =
  i\theta\ket{G}
  +
  \sum_{a,\alpha}T_{a\alpha}\ket{G_{a\alpha}},
\end{equation}
and the $\ket{G_{ a\alpha}}$ furnish a basis for the tangent space. \par 

As a consistency check, we may compare the dimension of the tangent space to the dimension of the manifold of Slater determinants. A Slater determinant is specified by an $n$-dimensional subspace inside the $m$-dimensional single-particle Hilbert space. The former can be specified by (the first $n$ columns of) a $U(m)$ unitary, but rotations within the subspace or its complement do not change the state, and hence we must mod out by $U(n)$ and $U(m-n)$ rotations. Since the real dimension of $U(d)$ is $d^2$, the real dimension of the Slater manifold (i.e. Grassmannian) is $m^2-n^2-(m-n)^2 = 2n(m-n)$ and the complex dimension is $n(m-n)$. Likewise, $T_{a\alpha}$ is an $n\times (m-n)$ complex matrix, and hence the dimensions match. 
\par 
Now that we have identified the tangent  directions at $\ket{G}$, we may derive the first variation of $\mathrm{F}(G)$ by considering a path of the form
\begin{equation}
  \ket{G(\tau)}
  =
  \ket{G}
  +
  \tau\sum_{a,\alpha}T_{a\alpha}\ket{G_{a\alpha}}
  +
  O(\tau^2).
\end{equation}
Defining 
$
  B_{a\alpha}=\langle G_{a\alpha}|\psi\rangle,
$
we have
\begin{equation}
  \langle G(\tau)|\psi\rangle
  =
  \langle G|\psi\rangle
  +
  \tau\sum_{a,\alpha}\overline{T_{a\alpha}}\,B_{a\alpha}
  +
  O(\tau^2),
\end{equation}
and therefore
\begin{align}
  \mathrm F(G(\tau))
  &=
  |\langle G(\tau)|\psi\rangle|^2 \nonumber\\
  &=
  |\langle G|\psi\rangle|^2
  +
  2\tau\operatorname{Re}
  \left[
     \langle \psi|G\rangle
    \sum_{a,\alpha}\overline{T_{a\alpha}}\,B_{a\alpha}
  \right]
  +
  O(\tau^2).
\end{align}
Since each
$T_{a\alpha}$ is arbitrary, the first variation vanishes if and only if
\begin{equation}
    \langle \psi|G\rangle B_{a\alpha}=0
\end{equation}
for every $a,\alpha$. Hence, if $\mathrm F(G)>0$, then $\langle G|\psi\rangle\ne0$, and this reduces to
\begin{equation}
  \langle G_{a\alpha}|\psi\rangle=0
  \qquad
  \forall a,\alpha
\end{equation}
which is equivalent to $B_G^2=
  \sum_{a,\alpha}|\langle G_{a\alpha}|\psi\rangle|^2
  =
  0$, as was to be shown.
\end{proof}

\begin{lemma}[Bounding higher particle-hole excitations]
\label{lem:tail}
Let $\ket{G}$ and $\ket{H}$ be normalized $n$-particle Slater
determinants. Let
\begin{equation}
  |G\rangle=c_{g_1}^\dagger\cdots c_{g_n}^\dagger|0\rangle
\end{equation}
have occupied orbitals $g_1,\ldots, g_n$ and
unoccupied orbitals $\{f_a\}$ with corresponding creation operators $c_{f_a}^\dagger$, for $a=1,\ldots, m-n$. 
Decompose $\ket{H}$ into particle-hole sectors relative to $\ket{G}$:
\begin{equation}
  \ket{H}=\ket{h_0}+\ket{h_1}+\ket{h_{\geq 2}},
\end{equation}
where $\ket{h_0}$ is proportional to $\ket{G}$, $\ket{h_1}$ lies in the
span of the single particle-hole excitations $c_{f_a}^\dagger c_{g_\alpha}\ket{G}$ and $\ket{h_{\geq 2}}$ contains all two-and-higher particle-hole sectors.
Write
\begin{equation}
  s=\norm{\ket{h_0}}=|\langle G|H\rangle|,
  \qquad
  q=\norm{\ket{h_{\geq 2}}}.
\end{equation}
Then
\begin{equation}
  q\leq \sqrt{2}(1-s).
\end{equation}
\end{lemma}

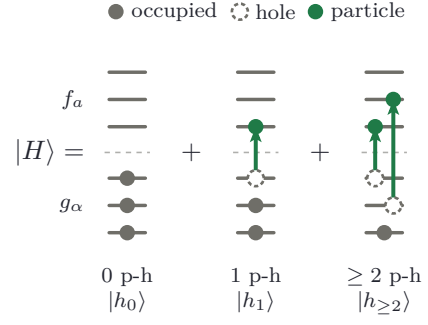
\begin{figure}
\centering
\definecolor{greenacc}{HTML}{1F7A47}
\definecolor{graytone}{HTML}{6F6E68}
\definecolor{cliquecol}{HTML}{23252F}
\begin{tikzpicture}[
  font=\footnotesize,
  lvl/.style   ={graytone, line width=1.3pt, line cap=round},
  fermi/.style ={graytone, line width=0.6pt, dash pattern=on 2pt off 2pt, opacity=0.55},
  occ/.style   ={fill=graytone, draw=none},
  hole/.style  ={fill=white, draw=graytone, line width=1pt, dash pattern=on 1.6pt off 1.1pt},
  part/.style  ={fill=greenacc, draw=none},
  exc/.style   ={-{Stealth[length=2mm,inset=0.4mm]}, greenacc, line width=1.5pt},
  lbl/.style   ={cliquecol},
]

\fill[occ]  (1.30,5.30) circle (2.9pt);
\node[lbl,anchor=west] at (1.39,5.30) {occupied};
\draw[hole] (2.94,5.30) circle (3.0pt);
\node[lbl,anchor=west] at (3.03,5.30) {hole};
\fill[part] (3.94,5.30) circle (2.9pt);
\node[lbl,anchor=west] at (4.03,5.30) {particle};

\node[lbl,font=\normalsize] at (0.42,3.44) {$\ket{H}=$};
\node[lbl,font=\normalsize] at (2.30,3.44) {$+$};
\node[lbl,font=\normalsize] at (4.00,3.44) {$+$};

\node[lbl,anchor=east] at (1.00,4.14) {$f_a$};
\node[lbl,anchor=east] at (1.00,2.74) {$g_\alpha$};

\foreach \y in {4.50,4.14,3.78,3.10,2.74,2.38}
  {\draw[lvl] (1.21,\y)--(1.69,\y);}
\draw[fermi] (1.15,3.44)--(1.75,3.44);
\foreach \y in {3.10,2.74,2.38}{\fill[occ] (1.45,\y) circle (2.9pt);}
\node[lbl,anchor=north] at (1.45,2.04) {$0$ p-h};
\node[lbl,anchor=north] at (1.45,1.72) {$\ket{h_0}$};

\foreach \y in {4.50,4.14,3.78,3.10,2.74,2.38}
  {\draw[lvl] (2.91,\y)--(3.39,\y);}
\draw[fermi] (2.85,3.44)--(3.45,3.44);
\draw[exc] (3.15,3.10)--(3.15,3.78);   
\foreach \y in {2.74,2.38}{\fill[occ] (3.15,\y) circle (2.9pt);}
\draw[hole] (3.15,3.10) circle (3.0pt);
\fill[part] (3.15,3.78) circle (2.9pt);
\node[lbl,anchor=north] at (3.15,2.04) {$1$ p-h};
\node[lbl,anchor=north] at (3.15,1.72) {$\ket{h_1}$};

\foreach \y in {4.50,4.14,3.78,3.10,2.74,2.38}
  {\draw[lvl] (4.61,\y)--(5.09,\y);}
\draw[fermi] (4.55,3.44)--(5.15,3.44);
\draw[exc] (4.73,3.10)--(4.73,3.78);   
\draw[exc] (4.97,2.74)--(4.97,4.14);
\fill[occ] (4.85,2.38) circle (2.9pt);
\draw[hole] (4.73,3.10) circle (3.0pt);
\draw[hole] (4.97,2.74) circle (3.0pt);
\fill[part] (4.73,3.78) circle (2.9pt);
\fill[part] (4.97,4.14) circle (2.9pt);
\node[lbl,anchor=north] at (4.85,2.04) {$\ge 2$ p-h};
\node[lbl,anchor=north] at (4.85,1.72) {$\ket{h_{\ge2}}$};
\end{tikzpicture}
\caption{\textbf{Particle-hole sectors.} Relative to a reference Slater
determinant $\ket{G}$, any Slater $\ket{H}$ decomposes orthogonally into
particle-hole sectors $\ket{H}=\ket{h_0}+\ket{h_1}+\ket{h_{\ge2}}$ graded by the
number of particle-hole (p-h) excitations. The
zero sector is proportional to $\ket{G}$, the 1 p-h
sector spans the tangent space of the Slater manifold at $\ket{G}$, and the
remaining term $\ket{h_{\ge2}}$ comprises 2 or more p-h excitations, with norm controlled by
Lemma~\ref{lem:tail}.}
\label{fig:phsectors}
\end{figure}
\begin{proof}
This proof borrows some mathematical steps from Ref.~\cite{bakshi2025learning}, though we present it in a self-contained manner. If $s=0$, then $q\leq 1\leq \sqrt{2}$, and the result is immediate, so we assume $s>0$ in the following. We write
\begin{equation}
\ket{G}=c_{g_1}^\dagger\cdots c_{g_n}^\dagger\ket{0},
  \qquad  \ket{H}=c_{h_1}^\dagger\cdots c_{h_n}^\dagger\ket{0},
\end{equation}
where $g_1,\ldots,g_n$ and $h_1,\ldots,h_n$ are orthonormal occupied
orbitals for $\ket{G}$ and $\ket{H}$, respectively, and define the overlap matrix
\begin{equation}
  M^{(G|H)}_{\alpha\beta}=\langle g_\alpha|h_\beta\rangle.
\end{equation}
Taking a singular value decomposition of $M^{(G|H)}$ and rotating the occupied
bases of $\ket{G}$ and $\ket{H}$ accordingly, we may assume that
\begin{equation}
  M^{(G|H)}_{\alpha\beta}
=\sigma_\alpha\delta_{\alpha\beta}
\end{equation}
with singular values $
  0\leq \sigma_\alpha\leq 1$. Since
\begin{equation}
  s=|\langle G|H\rangle|=|\det M^{(G|H)}|=\prod_{\alpha=1}^n\sigma_\alpha
\end{equation}
and $s>0$, all $\sigma_\alpha$ are nonzero. We write
\begin{equation}
\sigma_\alpha=\cos\theta_\alpha,
  \qquad
  0\leq \theta_\alpha < \frac{\pi}{2}.
\end{equation}
For each $\alpha$ with $\theta_\alpha>0$, we may define an unoccupied
orbital of $\ket{G}$ by
\begin{equation}
  f_\alpha
  =
  \frac{h_\alpha-\cos\theta_\alpha\,g_\alpha}{\sin\theta_\alpha}.
\end{equation}
If $\sin\theta_\alpha=0$, the choice of $f_\alpha$ will be irrelevant. The
orbitals $f_\alpha$ are orthonormal and orthogonal to the occupied
orbitals of $\ket{G}$. Thus the occupied orbitals of $\ket{H}$ can be
written as $  h_\alpha
  =
  \cos\theta_\alpha\,g_\alpha
  +
  \sin\theta_\alpha\,f_\alpha.$ Equivalently,
\begin{equation}
  \ket{H}
  =
  \prod_{\alpha=1}^n
  \left(
    \cos\theta_\alpha\,c_{g_\alpha}^\dagger
    +
    \sin\theta_\alpha\,c_{f_\alpha}^\dagger
  \right)\ket{0}.
\end{equation}

We now expand this product to obtain a particle-hole decomposition of $\ket{H}$ relative to $\ket{G}$ (as in Fig.~\ref{fig:phsectors}). First, the term with no $f$ orbitals is precisely $\ket{h_0}$ and satisfies
\begin{equation}
  \norm{\ket{h_0}}^2
  =
  s^2
  =
  \prod_{\alpha=1}^n \cos^2\theta_\alpha.
\end{equation}
Next, the terms with exactly one $f$ orbital add up to give $\ket{h_1}$, and hence we have
\begin{equation}
  \norm{\ket{h_1}}^2
  =
  \sum_{\alpha=1}^n
  \sin^2\theta_\alpha
  \prod_{\beta\ne\alpha}\cos^2\theta_\beta.
\end{equation}
It is useful to set $r_\alpha=\tan^2\theta_\alpha.$
Then
\begin{equation}
  \norm{\ket{h_1}}^2
  =
  s^2\sum_{\alpha=1}^n r_\alpha.
\end{equation}
Finally, all remaining terms add up to give $\ket{h_{\geq 2}}$. Since $\ket{H}$ is normalized, we have
\begin{align}
  q^2
  &=
  \norm{\ket{h_{\geq2}}}^2 =
  1-\norm{\ket{h_0}}^2-\norm{\ket{h_1}}^2 \nonumber\\
  &=
  1-s^2\left(1+\sum_{\alpha=1}^n r_\alpha\right).
\end{align}

Next, we lower bound $\sum_\alpha r_\alpha$ in terms of $s$. For
$r\geq0$, we have $r\geq \log(1+r)$. Hence
\begin{equation}
  \sum_{\alpha=1}^n r_\alpha
  \geq
  \sum_{\alpha=1}^n \log(1+r_\alpha).
\end{equation}
Using $ s^2
  =
  \prod_{\alpha=1}^n \cos^2\theta_\alpha
  =
  \prod_{\alpha=1}^n (1+r_\alpha)^{-1},$ 
we get
\begin{equation}
  \sum_{\alpha=1}^n \log(1+r_\alpha)
  =
  -\log s^2
  =
  -2\log s
\end{equation}
and therefore $  q^2 \leq 1-s^2(1-2\log s).$ It remains to show that the right-hand side is bounded by
$2(1-s)^2$. Define
\begin{equation}
  g(s)
  =
  2(1-s)^2-(1-s^2(1-2\log s)).
\end{equation}
 We claim $g(s)\geq 0$ on the interval $0<s\leq 1$. For an elementary proof, note that $g(1)=0$ and $g'(s)=4(s-1-s\log s)\leq 0$ on this interval. Since $g$ decreases as $s$ increases to
$1$ and $g(1)=0$, we have $ g(s)\geq0$. Hence $1-s^2(1-2\log s)\leq 2(1-s)^2,$
which combines with the previous bound to give us 
\begin{equation}
  q^2\leq 2(1-s)^2.
\end{equation}
Taking square roots proves the claim. 
\end{proof}
\begin{lemma}[Agnostic quantum threshold search]
    \label{lem:agnosticQTS}
Let $A_1,\ldots, A_M$ be known observables with $0\preceq A_i \preceq 1$ and let $F_i \coloneqq \Tr(\rho A_i)$ and $F_\star = \max_i F_i$. Given quantum copies of $\rho$, there is a procedure which outputs an index $i_\star$ such that $F_{i_\star} \geq F_\star -\varepsilon$ with probability at least $1-\delta$, using 
\begin{equation}
    n_{\mathrm{AQTS}} = O\left(\frac{\mu}{\varepsilon^2}\left(\log^2M + \log\frac{\mu}{\delta}\right) \log\frac{\mu}{\delta}\right)
\end{equation}
copies of $\rho$, where $\mu \coloneqq \log \frac{1}{\varepsilon}$.
\end{lemma}
\begin{proof}
This result is a simple extension of the quantum threshold search (QTS) algorithm of Ref.~\cite{buadescu2021improved} where we run a binary search over the ``threshold variable" $\theta$. We refer to Section~\ref{sec:quantum} for a brief review of the quantum threshold search algorithm. \par 

Define $\xi = \varepsilon/12$ and $R = \lceil \log_2(12/\varepsilon)\rceil+1$. Let $\delta/R$ be the tolerated failure probability of each call to the QTS algorithm. \par 
We will initialize variables $L = 0$, $U=1$ and $i_{\star}=1$ and update them after each step. At each step, $L,U \in [0,1]$ are lower and upper bounds, respectively, on $F_{i_\star}$ and $F_\star$. \par 

At each step, we set $\theta = (L+U)/2$ and run QTS on $A_1,\ldots, A_M$ with a common threshold $\theta$ and accuracy $\xi$. If QTS returns an index $j$, then $F_j > \theta-\xi$. If moreover $\theta-\xi>L$, we update
\begin{equation}
    i_{\star} \leftarrow j, L \leftarrow \theta-\xi
\end{equation}
and otherwise leave $i_{\star}$ and $L$ unchanged. If QTS returns ``none", we update
\begin{equation}
    U \leftarrow \theta. 
\end{equation}
The size of the interval $[L,U]$ shrinks as $ U-L \mapsto \frac{U-L}{2}+\xi$. After $R$ steps, $U-L \leq 2\xi + 2^{-R} \leq \varepsilon/3$, and in the event of no failures we have
\begin{equation}
    F_{i_\star} \geq L \geq U-\varepsilon/3 \geq F_{\star}-\varepsilon/3,
\end{equation}
which is stronger than the required bound. From a union bound over the failure probabilities it follows that the total failure probability is at most $\delta$. The stated value for $n_{\mathrm{AQTS}}$ follows immediately from $n_{\mathrm{QTS}}$ (c.f. Eq.~\eqref{eq:nQTS}) and our choice of $R$. 
\end{proof}

\section{Details of numerics}
\label{app:numerics}

In this section, we collect the details of the numerical plots in the main text. The general aim of the numerics shown in this paper is to establish that finding the closest Slater using our classical algorithm (classical) can still be a useful and meaningful task, at least for a small number of particles. Our algorithm also provides a new benchmark against which we can definitively assess the performance of heuristic closest Slater algorithms such as stochastic gradient ascent. \par 

In Fig.~\ref{fig:1}c, we study the Fermi-Hubbard model~\cite{arovas2022hubbard}
\begin{equation}
    H \;=\; -t\sum_{\langle ij\rangle,\sigma} c^\dagger_{i\sigma}c_{j\sigma} \;+\; U\sum_i n_{i\uparrow}n_{i\downarrow}
\end{equation}
on a $4\times 4$ square lattice with cylindrical boundary conditions (open along $\hat x$, periodic along $\hat y$) at a fixed electron filling $n_\uparrow = n_\downarrow = 2$ and we sweep over the coupling $|U|/t\in [0,5]$ for both signs of $U$. We produce the ground state by exact diagonalization. We compute and diagonalize the 1-RDM $\gamma_{(i,\sigma),(i',\sigma')} = \langle c^\dagger_{i'\sigma'}c_{i\sigma}\rangle$ to find the active orbitals with a fidelity tolerance $\varepsilon=0.02$. We parameterize the manifold of Slaters over the active space as follows. Writing $n=n_\uparrow+n_\downarrow$ and $r$ for the dimension of the active space, we fix an orthonormal frame as $W_0=\begin{pmatrix}I_n \\ 0\end{pmatrix}\in\mathbb{C}^{r\times n}$. Other $n$-planes are specified by an $(r-n)\times n$ matrix $Z$, by setting $\Omega(Z)=\begin{pmatrix}0&-Z^\dagger\\ Z&0\end{pmatrix}$, an antihermitian matrix, so $\exp\Omega(Z)$ is unitary, and then setting
\begin{equation}
W(Z)=\exp[\Omega(Z)]W_0,
\end{equation}
whose $n$ columns define another $r\times n$ orthonormal frame. We can restrict each entry of $Z$ to $[-\pi/2,\pi/2]\oplus i[-\pi/2,\pi/2]$ and cover the entire manifold of Slaters~\cite{edelman1998geometry}. This gives a compact parameterization with domain $[-\pi/2,\pi/2]^{2(r-n)\times n}$, over which we form a covering net and perform an explicit search to evaluate $\OPT$. \par

In Fig.~\ref{fig:numerics}(a) we repeat this study with the same ground states represented by a neural quantum state (NQS), a multi-determinant Slater-Jastrow ans\"atz with four determinants and a neural Jastrow factor, optimized by variational Monte Carlo with stochastic
reconfiguration~\cite{vicentini2022netket}. From the trained NQS we run the search
using only samples $x\sim|\psi(x)|^2$. We estimate the 1-RDM, isolate the active space at
tolerance $\varepsilon=0.02$ as above, and evaluate the fidelity of each trial
Slater $\ket{S}$ with the estimator
$\widehat{F} = |\overline{r}|^2 / \overline{|r|^2}$, where
$r(x)=\langle x|S\rangle/\psi(x)$ and the averages are taken over a fixed sample
set. Dashed lines are the exact ED from Fig.~\ref{fig:1}c and
the red band is the $\OPT$ for the NQS, with width representing Monte-Carlo standard
deviation over independent sample sets. This confirms the viability of computing $\OPT$ using the closest Slater algorithm for modest system sizes using sample access to a neural quantum state representation. We note that the small upward drift at strong
attractive coupling reflects the imperfect NQS approximation of the ground state.

\par

In Fig.~\ref{fig:numerics}(b) we use $\OPT$ as the ground truth against
which to judge heuristic optimization. For Fermi--Hubbard ground states at fixed filling
$n_\uparrow=n_\downarrow=2$ and increasing system size $m$, we run randomly initialized local gradient ascent on the Slater manifold
and record the fraction of starts whose stationary fidelity reaches $\OPT$
(within $5\times10^{-3}$). We fix $U/t=2$ and compare ascent step budgets of $50$, $300$,
and $2000$ iterations. For every budget the success rate falls sharply with $m$ (to a few
percent or less by $m=64$) and a larger budget shifts the curves upward without changing
this trend. We note that a
random Slater has fidelity of order one over the many body Hilbert space dimension with the target state, placing it on an exponentially
flat region of the landscape where the gradient nearly vanishes and the ascent stalls. The search, by contrast, attains $\OPT$ at every size. Therefore, randomly initialized local optimization may not be a reliable route to the
closest Slater problem at scale, motivating our protocol.

\begin{figure}
    \centering    \includegraphics[width=\linewidth]{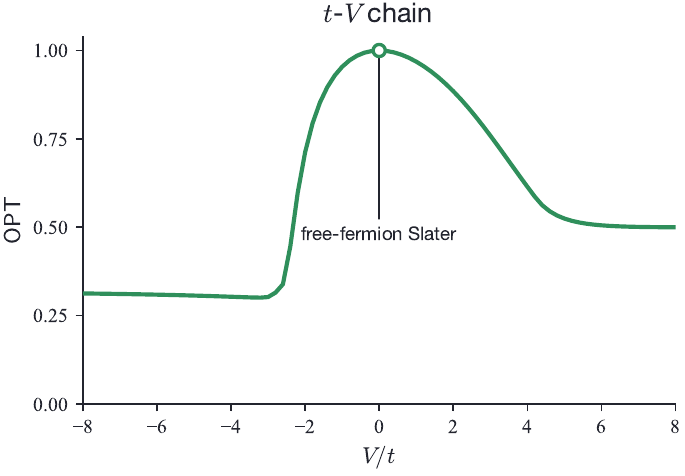}
\caption{\textbf{Closest Slater, $t$-$V$ chain.} Closest-Slater
fidelity $\OPT$ of the ground state of the half-filled spinless $t$-$V$ chain
[Eq.~(\ref{eq:tV}), $m=10$, no disorder, periodic boundary conditions] versus the nearest-neighbor interaction $V/t$.
At $V=0$ the free-fermion ground state is a single Slater ($\OPT=1$); interactions
drive it away from any Slater in both directions---toward $\OPT\approx 1/2$ for strong
repulsion, where the state approaches the charge-density-wave two-crystal cat state and to $\OPT\approx 0.3$ in the
attractive regime.}
\label{fig:tvopt}
\end{figure}

\par
In Fig.~\ref{fig:numerics}(c) we probe the optimization landscape of the
closest Slater problem for several natural ensembles of pure and mixed states,
confirming the $\mathrm{F}=2/3$ threshold of Sec.~\ref{sec:highfid}. We work in
the $n$-particle sector over $m=2n$ fermionic modes, whose occupation basis
$\{\ket{x}\}$ of Slater determinants has dimension $D=\binom{m}{n}$. For a state
$\rho$ the fidelity of a Slater is $\mathrm{F}(S)=\langle S|\rho|S\rangle$, with
$\OPT=\max_{\ket{S}}\mathrm{F}(S)$, and a \textit{spurious stationary point} is a
stationary Slater with $\mathrm{F}(S)<\OPT$. For each state we compute $\OPT$, run many random-initialized local ascents to convergence, and
collect the fidelities of the spurious stationary points they reach. Each state contributes a horizontal tick at $\OPT$ and a point at
every spurious level found. We sample four ensembles at $n=5$, $m=10$. In the
\textit{Haar-random} ensemble the state is pure with i.i.d.\ Gaussian coefficients
$\psi_x$. In the \textit{weakly correlated} ensemble it is a determinant ``cat'': a
superposition of three to six Slater determinants at pairwise Hamming distance $\geq 2$
with comparable random weights. The \textit{random mixed} ensemble is a rank-two mixture
$\rho=p\ket{v_1}\!\bra{v_1}+(1-p)\ket{v_2}\!\bra{v_2}$ of two independent Haar-random pure
states, $p\sim\mathrm{U}[0.35,0.65]$. The
\textit{physical} ensemble consists of ground states of the spinless $t$--$V$ chain
\begin{equation}\label{eq:tV}
H= -t\sum_{\langle ii'\rangle}\! c_i^\dagger c_{i'}
         + V\sum_{i=1}^{m} n_i\,n_{i+1}
         + \sum_{i=1}^{m}\varepsilon_i\,n_i ,
\end{equation}
spinless fermions on a periodic ring of $m=10$ sites ($c_{m+1}\equiv c_1$,
$n_i=c_i^\dagger c_i$) at half filling $n=5$, with hopping $t=1$, nearest-neighbor
interaction $V$, and on-site potentials $\varepsilon_i$. We refer to Fig.~\ref{fig:tvopt} for a plot of $\OPT$ in this model without disorder and with varying $V/t$. For the figure in the main text, we draw $V\sim\mathrm{U}[-8,8]$
and, in $60\%$ of instances, weak disorder
$\varepsilon_i\sim\mathcal{N}(0,w^2)$ with $w\sim\mathrm{U}[0.1,0.9]$ (otherwise
$\varepsilon_i=0$). A fifth, \textit{adversarial} column contains states from the
sharpness construction of Sec.~\ref{sec:highfid}, whose planted spurious maximum sits at
fidelity $\alpha^2$ approaching $2/3$; extending to $n=16$ reaches a spurious fidelity of
$0.642$. Across every ensemble all spurious stationary points lie strictly
below $2/3$, consistent with Theorem~\ref{thm:certificate} and its mixed-state extension.
For the natural ensembles they concentrate well below the threshold---at most $\approx
0.13$, $0.19$, $0.39$, and $0.49$ for the mixed, Haar, weakly correlated, and physical
ensembles respectively---so the bound is conservative there. Physical ground states rarely
admit any spurious point; the physical column additionally includes verified
charge-density-wave ``cat'' states which likewise remain
below $2/3$. Only the adversarial states approach the threshold, confirming that
$2/3$ is optimal in the worst case.

\end{document}